\documentclass[11pt]{article}
\usepackage[utf8]{inputenc}
\usepackage{fullpage,times}
\usepackage{amsmath, amssymb, amsthm}
\usepackage[hidelinks]{hyperref}
\usepackage{cleveref}
\usepackage{multicol}
\usepackage{color}
\usepackage[shortlabels]{enumitem}
\usepackage{xspace}
\usepackage{comment}
\newtheorem{thm}{Theorem}[section]

\newtheorem{lemma}[thm]{Lemma}
\newtheorem{theorem}[thm]{Theorem}
\newtheorem{fact}[thm]{Fact}
\newtheorem{corollary}[thm]{Corollary}
\newtheorem{question}[thm]{Question}
\newtheorem{definition}[thm]{Definition}

\theoremstyle{definition}

\newcommand{\eps}{\varepsilon}
\newcommand{\R}{\mathbb{R}}
\newcommand{\bone}{\mathbf{1}}
\newcommand{\Paren}[1]{\left(#1\right)}

\newcommand{\N}{\mathbb{N}}
\newcommand{\optis}{\mathrm{OPT}_{\mathrm{IS}}}
\newcommand{\tP}{\tilde{P}}

\newcommand{\polylog}{\mathop{\rm polylog}}
\newcommand{\poly}{\mathop{\rm poly}}
\newcommand{\MAXCOV}{\textsc{Max-Coverage}}
\newcommand{\OPTest}{\OPT_{\text{est}}}
\newcommand{\OPT}{\text{OPT}}
\newcommand{\SSS}{\mathcal{S}}
\newcommand{\UUU}{\mathcal{P}}
\newcommand{\OPTLP}{\text{OPT}_{\text{LP}}}
\newcommand{\OPTDISLP}{\text{OPT}^{\text{DIS}}_{\text{LP}}}
\newcommand{\Ex}{\mathbb{E}}
\newcommand{\Pbad}{P_{\text{bad}}}
\newcommand{\Pgood}{P_{\text{good}}}
\newcommand{\Sbad}{S_{\text{bad}}}
\newcommand{\Sgood}{S_{\text{good}}}
\newcommand{\xxx}{\widetilde{x}}
\newcommand{\vmax}{v_{\text{max}}}
\newcommand{\vol}[1]{\text{vol}\left(#1\right)}

\iftrue
\newcommand{\pasin}[1]{\textcolor{red}{Pasin: #1}}
\else
\newcommand{\pasin}[1]{}
\fi

\title{Approximation Algorithms for Geometric Maximum Coverage}

\author{%
  Sujoy Bhore\thanks{IIT Bombay. 
    \texttt{sujoy@cse.iitb.ac.in}}
  \and
  Timothy M. Chan\thanks{University of Illinois at Urbana-Champaign.
    \texttt{tmc@illinois.edu}}
  \and
  Pasin Manurangsi\thanks{Google Research.
    \texttt{pasin@google.com}}
}

\date{\today}

\begin{document}

\maketitle
\setcounter{page}{0}
\thispagestyle{empty}

\author{}
\begin{abstract}
We study the \emph{maximum coverage} problem for geometric set systems: given a set of points, a set of geometric objects, and a number $k$, select $k$ objects maximizing the number of points inside their union.

\begin{enumerate}
\item We present a polynomial-time approximation algorithm, via LP rounding, achieving approximation factor strictly better than $1-1/e$ for any set system with linear \emph{2-shallow cell complexity} (or any set system that can be decomposed into a constant number of such set systems).  The result also holds for the \emph{weighted} (or \emph{budgeted}) maximum coverage problem, where objects have weights and we want to select objects with total weight not exceeding a given budget.  The result applies to many types of geometric objects, including pseudodisks in 2D, fat axis-aligned rectangles in 2D, similar-size fat triangles in 2D, axis-aligned unit cubes in 3D, etc.---in fact, these are the same types of objects for which LP-based, $O(1)$-approximation algorithms were known for the (weighted) \emph{geometric set cover} and \emph{discrete independent set} problems.

\item For small $k$, we obtain a $(1-\varepsilon)$-approximation algorithm more generally for any set system with \emph{constant VC dimension}, running in time exponential in $\tilde{O}(k/\varepsilon)$.  This simplifies and improves Badanidiyuru, Kleinberg, and Lee's previous parameterized approximation scheme [SoCG'12]
running in time exponential 
in
$\tilde{O}(k^2/\varepsilon^5)$.

\item A continuous version of the geometric maximum coverage problem asks for $k$ objects maximizing the volume of their union. 
We give better approximation algorithms for this problem for certain families of objects;
for example, we obtain an EPTAS for fat convex objects in any constant dimension.

\item We complement our algorithms with several hardness results, e.g.,
APX-hardness for fat axis-aligned rectangles in 2D and other
types of objects, $(1-1/e+\eps)$-approximation hardness for axis-aligned boxes in a dimension dependent on $\eps$,
and a lower bound ruling out $n^{\mathop{\rm poly}(1/\varepsilon)}$-time PTASs for the continuous problem for (non-fat) axis-aligned boxes in 3D.
\end{enumerate}
\end{abstract}

\newpage
\tableofcontents
\thispagestyle{empty}
\newpage

\setcounter{page}{1}
\section{Introduction}

The \emph{max coverage} problem is one of the most fundamental optimization problems in combinatorial optimization. In its classical form, we are given a collection $S$ of sets  over a universe $P$ together with an integer $k$, and the goal is to select at most $k$ sets from $S$ whose union has maximum cardinality. A simple greedy algorithm achieves an approximation factor of $1-1/e$~\cite{CornuejolsNW80, FisherNW79,Hochbaum97}, and so does LP rounding~\cite{AgeevS04}, while the celebrated result of Feige~\cite{Feige98} shows that this guarantee is optimal unless $\mathrm{P}=\mathrm{NP}$.

Weighted extensions of the problem can also be considered. Allowing elements in $P$ to have weights does not really make the problem more challenging: we can just treat weights in $P$
as multiplicities.\footnote{
Objects with weights smaller than $1/n^{c+1}$ times the largest weight may be removed, while worsening the approximation factor by only $1-O(1/n^c)$ for an arbitrary large constant $c$.  Then by rescaling, all weights become integers bounded by $n^{c+1}$.
}
More interesting is to allow the sets in $S$ to have (positive) weights: the goal is now to select sets from $S$ with total weight
at most a given budget whose union has maximum cardinality.  This version  is sometimes called the \emph{budgeted max coverage} problem in the literature, although we will refer to it simply as the \emph{weighted} version of max coverage.
 A variant of the greedy algorithm~\cite{KhullerMN99} or LP rounding~\cite{Srinivasan01,DoerrKW10} is also known to achieve approximation factor of $1-1/e$ for this weighted problem.

Over the past three decades, max coverage has become a central primitive across approximation algorithms, submodular optimization, machine learning, and computational geometry, among others. A major research direction has been to understand whether additional structure in the input allows one to circumvent the general $1-1/e$ barrier~\cite{BadanidiyuruKL12,ChaplickDRS18, JainKPSSSU26}.
Geometric set systems provide one of the most prominent examples,
a common setting of which is when $P$ is a given finite set of points and 
the sets of $S$ are geometric objects such as disks, rectangles, etc.\ (or, 
more precisely, a set in $S$ consists of all points of $P$ inside an object).

The main question addressed in this paper is the following:

\begin{quote}
\emph{When does geometry help for the (unweighted or weighted) max coverage problem? In other words, for what types of geometric objects can we obtain approximation guarantees strictly better than $1-1/e$?}
\end{quote}

\paragraph{Prior work on geometric set cover.}
To better understand the context of this question, we begin by reviewing
prior work on the closely related, and equally fundamental, optimization problem:
\emph{set cover}.  The goal is to select the smallest number of sets in $S$ so that the 
union covers $P$ (the number $k$ is no longer part of the input).  The standard greedy algorithm or randomized LP rounding achieves an approximation factor of
$\ln |P| + O(1)$~\cite{Hochbaum97,VazBOOK}, and 
Feige's seminal work~\cite{Feige98} showed that a $(1-o(1))\ln |P|$ approximation is the best possible for general set systems unless $\mathrm{P}=\mathrm{NP}$. 

Similarly, in the weighted set cover problem, the goal is to select sets in $S$ minimizing
their total weight, so that the union covers $P$.  LP rounding~\cite{Hochbaum97,VazBOOK} also achieves approximation 
factor of $\ln |P| + O(1)$ for the weighted problem.

\begin{table}\centering\small
\begin{tabular}{ccc}
objects & \begin{tabular}{c}approx.\ factor of\\poly-time algs.\end{tabular} & \!\!hardness\!\! \\\noalign{\smallskip}\hline\noalign{\bigskip}
\begin{minipage}{.5\textwidth}
2D unit disks, 2D disks, 2D convex pseudodisks,
2D unit squares, 2D squares,
2D unit-height rectangles,
2D translates/homothets of one convex object,
3D halfspaces
\end{minipage}
 & \begin{tabular}{c}
   $1+\eps$ for unweighted \cite{MustafaR10}\\
   $O(1)$ for weighted \cite{ChanGKS12}
   \end{tabular} & NP-hard\\\noalign{\medskip}\hline\noalign{\medskip}
\begin{minipage}{.5\textwidth}
2D pseudodisks, 2D $r$-admissible objects,
3D unit cubes
\end{minipage}
  & \begin{tabular}{c}
  $O(1)$\ for unweighted~\cite{clarkson2005improved}\\
  and weighted~\cite{ChanGKS12}
  \end{tabular} &  NP-hard\\\noalign{\medskip}\hline\noalign{\medskip}
\begin{minipage}{.5\textwidth}
2D fat rectangles, 
2D similar-height rectangles,
2D fat rotated rectangles with $O(1)$ orientations,
2D vertical \& horizontal slabs,
2D translates/homothets of $O(1)$ convex objects,
2D similar-size fat triangles,
2D rectangles intersecting the $x$-axis,
3D similar-size fat boxes,
3D boxes containing the origin,
$O(1)$-D axis-aligned slabs
\end{minipage}
 & \begin{tabular}{c}
  $O(1)$\ for unweighted~\cite{clarkson2005improved}\\
  and weighted~\cite{ChanGKS12}
  \end{tabular} & \!\!\!\!\begin{tabular}{c}APX-hard\\\cite{ChanG14}\end{tabular}\!\!\!\!\\\noalign{\medskip}\hline\noalign{\medskip}
\begin{minipage}{.5\textwidth}
2D rectangles,  3D unit balls,
3D cubes, 4D unit hypercubes, 4D boxes with the origin as a vertex, 4D halfspaces
\end{minipage} 
  & \begin{tabular}{c}
   $O(\log\textrm{OPT})$ for unweighted \cite{bronnimann1994almost}\\
   $O(\log |P|)$ for weighted
   \end{tabular} 
  & \!\!\!\!\begin{tabular}{c}APX-hard\\\cite{ChanG14}\end{tabular}\!\!\!\!\\\noalign{\medskip}\hline
 \end{tabular}
 \caption{Known results for unweighted and weighted geometric set cover for various types of objects.}\label{tbl:set:cov}
\end{table}

Over the last few decades, there has been an extensive body of work on the set cover problem in geometric settings.
Here is a summary of the major techniques:
\begin{itemize}
\item \emph{Local search PTASes}: Mustafa and Ray~\cite{MustafaR10} obtained a PTAS via local search for the unweighted set cover problem for geometric objects that satisfy
a certain \emph{planar graph support} property.
For example, this includes the objects listed in the first row of Table~\ref{tbl:set:cov}%
\footnote{
Throughout this paper, all squares/rectangles/boxes are axis-aligned, unless we explicitly say ``rotated''.  ``Similar height/size'' means that the ratio of the maximum to minimum height/size is bounded by $O(1)$.
}
(see \cite{DeL23} on the convex pseudodisk case).
However, there are many types of objects that do not satisfy this property, for example,
those in the third and fourth row of Table~\ref{tbl:set:cov}); moreover, local search only
works for the unweighted problem.

\item \emph{Separator-based QPTASes}: Mustafa, Raman, and Ray~\cite{MustafaRR15} described
a different, separated-based approach that yields a QPTAS for some of the types of objects appearing in the first row of Table~\ref{tbl:set:cov}.  However, we are interested only in polynomial-time approximation algorithms in this paper, and so will ignore this kind of results.

\item \emph{LP-rounding algorithms}: 
Another class of approximation algorithms is via rounding of the standard set cover LP\@.  The analysis of the integrality gap is related to a combinatorial geometry concept called \emph{$\eps$-nets}.

Br\"onnimann and Goodrich~\cite{bronnimann1994almost} used this approach to obtain $O(\log\OPT)$-approximation
algorithms for unweighted set cover for any set system with constant \emph{VC dimension}, where $\OPT$ denotes the optimal value.  This result
is very general (it includes
most natural families of geometric objects of constant description complexity in constant dimensions), but the improvement from $\log |P|$ to $\log\OPT$
is marginal unless $\OPT$ is subpolynomially small.

Clarkson and Varadarajan~\cite{clarkson2005improved} 
(see also~\cite{AronovES10}) obtained $O(1)$-approximation algorithms
for unweighted geometric set cover for objects with linear \emph{union complexity}.

Finally, building on Varadarajan's \emph{quasi-uniform sampling} technique~\cite{varadarajan2010weighted},
Chan, Grant, K\"onemann, and Sharpe~\cite{ChanGKS12} generalized the result and obtained $O(1)$-approximation algorithms for weighted geometric set cover for objects that have 
\emph{$\ell$-shallow cell complexity} $f_\ell(n)\le n\ell^{O(1)}$.  
In most geometric examples, $\ell$-shallow cell complexity $f_\ell(n)\le n\ell^{O(1)}$ follows from
linear union complexity by using a standard Clarkson--Shor argument~\cite{clarkson1989applications}, but one advantage of using shallow cell complexity is that it can be defined abstractly for set systems without referring to geometry, unlike
union complexity, and so potentially may yield more applications outside of geometry.%
\footnote{For one recent example, Browne and Chang~\cite{BrowneC26} 
gave an application to a problem on planar graphs.}

Generalizing further, the $O(1)$-approximation algorithm also holds in the case when the input objects can be partitioned into $O(1)$ subsets each with bounded shallow cell complexity (e.g., see \cite{ChanD0SW18}).  This yields all the results in the second and third row of Table~\ref{tbl:set:cov}. 

\item \emph{APX-hardness}: On the negative side,
Chan and Grant~\cite{ChanG14} described a technique to prove APX-hardness for various types of objects, including all the ones in the third and fourth row of Table~\ref{tbl:set:cov}.
\end{itemize}

\paragraph{Prior work on geometric discrete independent set and other related problems.}
Another basic optimization problem, \emph{discrete independent set}, shares
a similar story: here, we want to find the largest (or, in the weighted case, the
largest-weight) subcollection of objects in $S$ so that no two chosen
objects contain a common point in $P$.  Known techniques similarly include:

\begin{itemize}
\item \emph{Local search PTASes}: 
Ene, Har-Peled, and Raichel~\cite{EneHR17} obtained a PTAS via local search for the unweighted discrete independent set problem for objects that satisfy
a planar graph support property similar to Mustafa and Ray's~\cite{MustafaR10}.
\item \emph{LP-rounding algorithms}: Chan and Har-Peled~\cite{ChanH12} gave $O(1)$-approximation
algorithms for unweighted and weighted discrete independent set for objects with linear union complexity.  In Appendix~\ref{sec:DIS}, we note that their analysis
actually extends to objects with linear 2-shallow cell complexity; and it similarly extends to the case when the objects can be partitioned into $O(1)$ subsets each with linear 2-shallow cell complexity.
\end{itemize}

LP-rounding approaches for geometric set cover have also been adapted
to solve a number of related problems, e.g., \emph{set multicover}~\cite{chekuri2012set,RamanR22},
\emph{partial set cover}~\cite{InamdarV18,ChekuriIQVZ22}, and \emph{online} variants of set cover~\cite{BhoreGK26}.  Partial set cover seems especially relevant, as it can be viewed as an ``inverse'' to max coverage (it asks for the minimum number of objects in $S$ to cover at least $k$
points in $P$); in fact, the two problems are equivalent for exact algorithms, though not for approximation algorithms.

\paragraph{Prior work on geometric max coverage.}
Returning to the max coverage problem, we see a far less complete
picture on when geometry helps, in contrast to the extensive literature on geometric set cover.
Badanidiyuru, Kleinberg, and Lee~\cite{BadanidiyuruKL12} studied the case when $k$ is small, which we will discuss more later (see also \cite{BergCH09,Jin00ZZ18} for more specialized results for small $k$).  For general $k$, the only major prior work is by Chaplick, De, Ravsky, and Spoerhase~\cite{ChaplickDRS18},
who obtained a local-search-based PTAS for unweighted max coverage for objects that satisfy a planar graph support property\footnote{
Chaplick et al.~\cite{ChaplickDRS18} required the ``support graph'' to have balanced sublinear-size separators instead of being planar.
} similar to Mustafa and Ray's~\cite{MustafaR10}.
This takes care of the types of objects in the first row of Table~\ref{tbl:max:cov}, but not the other types of objects.
Also, local search inherently does not work well for the weighted problem.
There has been no prior work on LP-rounding-based approaches that beat $1-1/e$ for geometric max coverage. The problem seems more challenging because the target range is narrower (we aim for a constant approximation factor above $1-1/e$), unlike minimization problems such as set cover, set multicover, or partial set cover (where any constant approximation factor would be an improvement over the general logarithmic bound).

\paragraph{Main new results.}
We present a new algorithm for geometric max coverage that has approximation factor $1-1/e+\Omega(1)$ for any family of objects with linear 2-shallow cell complexity (or input that can be partitioned into $O(1)$ such families).  The approximation factor is in fact relative to the optimal LP value, and the algorithm works for weighted problem as well (the approximation factor is not LP-relative in the weighted case though).
This yields 
the first result that beats $1-1/e$ for both the unweighted and weighted problem for all the types of objects listed in the second and third row of Table~\ref{tbl:max:cov}, as well as the first such result for the weighted problem for the types of objects in the first row.

Note that our result generalizes a known (non-geometric) result (e.g.,~\cite{AgeevS04}) that one can beat $1-1/e$ for max coverage for set systems with maximum set size $\Delta=O(1)$, since such set systems have linear 2-shallow cell complexity $f_2(n)=O(\Delta n)$.

Our algorithm is not complicated---it simply returns the best among 3 solutions:
(i) a standard LP-rounding solution, 
(ii) the greedy solution,
and (iii) a modified greedy solution in which we start with a discrete independent set.  The algorithm description requires less than half a page (see Section~\ref{sec:max:cov:lp}), and most of the work lies in its analysis, which requires a series of steps.

The quantatitive improvements over $1-1/e$ are admittedly small (we do not try to explicitly work out and optimize the precise approximation factors for specific types of objects), but our contributions are qualitatively interesting for at least two reasons. 
First, they unify the landscape for max coverage with set cover and discrete independent set:
geometry helps for max coverage on essentially the same set systems for which $O(1)$-approximation algorithms are known for set cover and discrete independent set.  In fact, our approach provides a general \emph{black-box reduction}, converting any LP-relative $O(1)$-approximation algorithm for discrete independent set into a $(1-1/e+\Omega(1))$-approximation algorithm for max coverage.  It is somewhat surprising that such a general reduction exists, and is to discrete independent set rather than set cover (Inamdar and Varadarajan~\cite{InamdarV18} and Chekuri et al.~\cite{ChekuriIQVZ22} described general reductions for LP-based algorithms from partial set cover to set cover).
Second, there has been empirical evidence that
methods like greedy and randomized LP rounding tend to achieve approximation factor better than $1-1/e$
(e.g., see \cite{DoerrKW10}).  Our result provides theoretical justification that variants of greedy and standard randomized LP rounding can indeed beat $1-1/e$ \emph{in the worst case} for many geometric instances.

\begin{table}\centering\small
\begin{tabular}{ccc}
objects & \begin{tabular}{c}approx.\ factor of\\poly-time algs.\end{tabular} & \!\!hardness\!\! \\\noalign{\smallskip}\hline\noalign{\bigskip}
\begin{minipage}{.5\textwidth}
2D unit disks, 2D disks, 2D convex pseudodisks,
2D unit squares, 2D squares,
2D unit-height rectangles,
2D translates/homothets of one convex object,
3D halfspaces
\end{minipage}
 & \begin{tabular}{c}
   $1-\eps$ for unweighted \cite{ChaplickDRS18}\\
   \!\!{\color{red}\bf\boldmath $1-1/e+\Omega(1)$ for weighted (Sec.~\ref{sec:max:cov:lp})}\!\!
   \end{tabular} & NP-hard\\\noalign{\medskip}\hline\noalign{\medskip}
\begin{minipage}{.5\textwidth}
2D pseudodisks, 2D $r$-admissible objects,
3D unit cubes
\end{minipage}
  & \begin{tabular}{c}
  \color{red}\bf\boldmath $1-1/e+\Omega(1)$ for unweighted\\
  \color{red}\bf  and weighted (Sec.~\ref{sec:max:cov:lp})
  \end{tabular} &  NP-hard\\\noalign{\medskip}\hline\noalign{\medskip}
\begin{minipage}{.5\textwidth}
2D fat rectangles, 
2D similar-height rectangles,
2D fat rotated rectangles with $O(1)$ orientations,
2D vertical \& horizontal slabs,
2D translates/homothets of $O(1)$ convex objects,
2D similar-size fat triangles,
2D rectangles intersecting the $x$-axis,
3D similar-size fat boxes,
3D boxes containing the origin,
$O(1)$-D axis-aligned slabs
\end{minipage}
 & \begin{tabular}{c}
  \color{red}\bf\boldmath $1-1/e+\Omega(1)$ for unweighted\\
  \color{red}\bf and weighted (Sec.~\ref{sec:max:cov:lp})
  \end{tabular} & \!\!\!\!\!\!\begin{tabular}{c}
  \color{red}\bf APX-hard\\ \color{red}\bf (Sec.~\ref{sec:hard:low})\end{tabular}\!\!\!\!\!\!\\\noalign{\medskip}\hline\noalign{\medskip}
\begin{minipage}{.5\textwidth}
2D rectangles, 3D unit balls,
3D cubes, 4D unit hypercubes, 4D boxes with the origin as a vertex, 4D halfspaces
\end{minipage} 
 & \begin{tabular}{c}
  $1-1/e$ for unweighted\\
  and weighted
  \end{tabular}
  & \!\!\!\!\!\!\begin{tabular}{c}APX-hard\\ 
  (\cite{BadanidiyuruKL12} \&\\ Sec.~\ref{sec:hard:low})\end{tabular}\!\!\!\!\!\!\\\noalign{\medskip}\hline
 \end{tabular}
 \caption{Results for unweighted and weighted geometric max coverage for various types of objects.  Results in {\color{red}\bf bold} are all new.}\label{tbl:max:cov}
\end{table}

\paragraph{More new results, on approximation schemes for small $k$.}
For general set systems, no FPT approximation scheme (FPT-AS) exists for max coverage, assuming Gap-ETH\footnote{The Gap Exponential Time Hypothesis (Gap-ETH)~\cite{Dinur16,ManurangsiR17} states that, for some constant $\delta > 0$, there is no $2^{o(n)}$-time algorithm that can distinguish between a satisfiable 3CNF formula and one which is not even $(1 - \delta)$-satisfiable, where $n$ denotes the number of variables.}; in fact, achieving $(1 - 1/e + \Omega(1))$-approximation is not even possible in FPT time~\cite{Cohen-AddadG0LL19,Manurangsi20}.
However, for set systems with bounded VC dimension~$d$,
Badanidiyuru, Kleinberg, and Lee~\cite{BadanidiyuruKL12} 
presented an FPT-AS parameterized by $k$, achieving
approximation factor $1-\eps$ with running time $2^{\tilde{O}(dk^2/\eps^5)}\poly(|P|,|S|)$
for any given $\eps>0$.
For constant $d$, this is a PTAS when $k$ is smaller than $\log^{0.5 - o(1)}|S|$.
The question of improving the running time was posed at the end of their paper.
We obtain a new time bound $2^{\tilde{O}(dk/\eps)}\poly(|P|,|S|)$,
which has not only better dependence in $k$ (we now get a PTAS when $k\leq \log^{1 - o(1)}|S|$), but also better dependence in $\eps$.
Moreover, our new algorithm and analysis are simpler (they can be described in 2 pages---see Section~\ref{sec:improve:BKL}).

\paragraph{More new results, on continuous max coverage (i.e., max-volume selection).}
One application of our results is to the \emph{continuous} version of the geometric max coverage problem:
here, $P=\R^d$, and cardinality is replaced by volume, i.e.,
the goal is to select $k$ objects from $S$ maximizing
the volume of the union.
This continuous problem, which we may call \emph{max-volume selection},
reduces to the discrete problem, by building
the arrangement of $S$, which has $O(n^d)$ cells, assuming that the objects have constant
description complexity, and replacing each cell with a point of weight proportional to
the volume of the cell.  (As mentioned, weights in $P$ may be treated as multiplicities.)
Thus, all our results for geometric (unweighted and weighted) max coverage automatically apply to the continuous counterparts.
The question is whether for the continuous version, we could beat $1-1/e$ for a wider class of geometric objects.

Bringmann, Cabello, and Emmerich~\cite{BringmannCE17} studied the (unweighted) max-volume selection problem and
described a PTAS for the case when the objects are axis-aligned boxes in $\R^d$ each
having the origin as a vertex (in contrast, the discrete version is APX-hard for $d\ge 4$,
as we show).  We give more positive results for (unweighted and weighted) max-volume selection:

\begin{itemize}
\item an EPTAS for fat convex objects (in particular, balls and hypercubes) in any constant dimension~$d$ which computes a $(1-\eps)$-approximation in $2^{O(1/\eps^{2d})}\poly(|S|)$ time, by using a shifted quadtree approach (see Section~\ref{sec:ptas}---in contrast, the discrete
version is APX-hard even for 2D fat rectangles);
\item a $(1-1/e+\Omega(1))$-approximation algorithm for objects that
are homothets of $O(1)$ convex objects in any constant dimension, via a modification of our LP-rounding method (see Appendix~\ref{sec:max:cov:var}---in contrast,
for the discrete version, we know how to beat $1-1/e$  for this case only in 2D). 
\end{itemize}

\paragraph{New results on hardness.}
There have been some prior results establishing
hardness for obtaining $(1-1/e+\eps)$-approximation
for geometric max coverage; e.g., see
\cite[Theorem~4.2]{ErlebachL08} on fat objects
and \cite[Theorem~1.2]{Cohen-AddadSL21} on set systems of bounded VC dimension.  However, these constructions are not realized by objects with low description complexity.
Badanidiyuru, Kleinberg, and Lee~\cite{BadanidiyuruKL12} proved
APX-hardness for max coverage for 2D rectangles and 4D halfspaces.  

We prove a number of new hardness results:
\begin{itemize}
\item We observe that Chan and Grant's APX-hardness
proofs~\cite{ChanG14} for geometric set cover can be easily modified to yield
APX-hardness for max coverage for all the types of objects in the third and fourth row in Table~\ref{tbl:max:cov} (see Section~\ref{sec:hard:low}).
\item For any constant $\eps>0$, we prove hardness for $(1-1/e+\eps)$-approximation
for max coverage for boxes in some constant dimension which is a function of $\eps$  (see Section~\ref{sec:hard:high}).
\item For the continuous problem, i.e., max-volume selection, we rule out
PTASes for boxes in high nonconstant dimensions (see Section~\ref{sec:hard:vol:high}).
\item For the max-volume selection problem, we also rule out PTASes with running time $n^{\poly(1/\eps)}$ (or, even more loosely, $n^{2^{O(1/\eps^{0.99})}}$)
for boxes in 3D  (see Section~\ref{sec:hard:vol:low}).
\end{itemize}

\section{Preliminaries}

We first recall the notion of shallow cell complexity, which was first formulated by Chan, Grant, K\"onemann, and Sharpe~\cite{ChanGKS12} and is rephrased below:

\begin{definition}
Let $\SSS$ be a class of objects over a universe $\UUU$.
For a set $S$ of objects in $\SSS$ and a point $p\in \UUU$, define the subset $S_{\mid p} = \{s\in S: p\in s\}$.
We say that $\SSS$ has \emph{$\ell$-shallow cell complexity $f_\ell(\cdot)$} if for every set $S\subseteq\SSS$ of objects, there are at most $f_\ell(|S|)$ 
distinct subsets in $\{S_{\mid p}: p\in\UUU,\ |S_{\mid p}|=\ell\}$.
\end{definition}

Our max coverage algorithm will use an algorithm for the discrete independent set problem, defined formally below:

\begin{definition}
In the \emph{max-weight discrete independent set (DIS)} problem, we are given a set $S\subseteq\SSS$ of objects and a set $P\subseteq\UUU$ of points, where each object $s$ has a weight $w_s>0$,
and we want to find a largest-weight subset $I\subseteq S$ such that every point of $P$ is contained in at most one object in $I$.

The standard LP relaxation for DIS is the following:
\[\begin{array}{lll}
\text{max}  & \displaystyle\sum_{s\in S} w_sy_s\\
\text{s.t.} & \displaystyle\sum_{s\in S:\,p\in s} y_s\le 1  & \forall p\in P\\
            &  0\le y_s\le 1                   & \forall s\in S.
\end{array}\]
\end{definition}

Chan and Har-Peled~\cite{ChanH12} gave 
$O(1)$-approximation algorithms for max(-weight) DIS for objects with linear union complexity via LP rounding.  (It is not important to know the precise definition of union complexity here.)
We note that their analysis only needs to assume
linear 2-shallow cell complexity.
For completeness, we include a full, concise
re-derivation in Appendix~\ref{sec:DIS}.



\begin{lemma}[Based on~\cite{ChanH12}]\label{lem:DIS}
Suppose $\SSS$ has 2-shallow cell complexity $f_2(n)\le cn$ for a constant $c$.
We can compute a feasible solution to the max-weight DIS problem of weight $\Omega(\OPTDISLP/c)$ in polynomial time, where $\OPTDISLP$ denotes the optimal LP value.
\end{lemma}

\newcommand{\PROOFDIS}{

In this appendix, we prove Lemma~\ref{lem:DIS} by reinterpretating Chan and Har-Peled's techniques for max-weight DIS~\cite{ChanH12} (more precisely, from their LP-based algorithm for unweighted DIS, and the conference version of their LP-based algorithm for weighted, non-discrete independent set, originally for objects with low union complexity).
Let $(y_s)_{s\in S}$ be the optimal LP solution.  

\paragraph{Algorithm.}  Let $b\ge 1$ be a parameter to be set later.
W.l.o.g., assume that the $w_s$'s are all distinct.

\begin{enumerate}
\item Pick a random subset $R\subseteq S$, where each object $s\in S$ is selected  independently
with probability $y_s/b$.
\item For each $s\in S$, put $s$ in $I$ iff (i) $s\in R$ and (ii) for all $s'\in S$ such that
$w_{s'}>w_s$ and $s\cap s'\cap P\neq\emptyset$, we have $s'\not\in R$.
Return $I$ (which is clearly a feasible solution).
\end{enumerate}

\paragraph{Analysis.} 
\begin{eqnarray*}
 \Ex\left[\sum_{s\in I}w_s\right] &=& \sum_{s\in S} w_s \cdot \frac{y_s}{b}  \prod_{\begin{subarray}{c}s'\in S:\\ w_{s'}>w_s,\ s\cap s'\cap P\neq\emptyset\end{subarray}} 
  \left( 1-\frac{y_{s'}}{b} \right)\\
 &\ge& \sum_{s\in S} \frac{w_s y_s}{b} \left( 1\ - \sum_{\begin{subarray}{c}s'\in S:\\ w_{s'}>w_s,\ s\cap s'\cap P\neq\emptyset\end{subarray}} 
 \delta y_{s'} \right)
  \ \ \ge\ \ \frac{\OPTDISLP}{b} \ -\ \frac{1}{b^2}\!\!\sum_{\begin{subarray}{c}s,s'\in S:\\ w_{s'}>w_s,\ s\cap s'\cap P\neq\emptyset\end{subarray}} 
   w_sy_sy_{s'}
\end{eqnarray*}
(by using the inequality $\prod_i (1-a_i) \ge (1-\sum_i a_i)$ for $a_i\in [0,1]$).

\newcommand{\RR}{\widehat{R}}
To bound this expression, we define another random subset $\RR$ where each object $s\in S$ is selected independently
with probability $y_s/2$, and we define an intermediate set of pairs for a given number $t$:
\[ \Psi_t \ =\ \{(s,s')\in \RR\times \RR: w_{s'}>w_s\ge t,\ \exists p\in s\cap s'\cap P\ \mbox{s.t.\ $p$ is not in any object of $\RR\setminus\{s,s'\}$}\}.
\]
By the 2-shallow cell complexity assumption, $|\Psi_t| \le c\,|\{s\in \RR: w_s\ge t\}|$, and so
\[ \Ex[|\Psi_t|] \ \le\ c\sum_{s\in S:\, w_s\ge t} y_s/2. \]
On the other hand, for a fixed  pair $(s,s')\in S\times S$ with $w_{s'} > w_s\ge t$ and $s\cap s'\cap P\neq\emptyset$,
we can fix a point $p_{s,s'}\in s\cap s'\cap P$ and infer (by a typical Clarkson--Shor-type argument~\cite{clarkson1989applications}) that 
\begin{eqnarray*}
\Pr[(s,s')\in \Psi_t] &\ge& \frac{y_s y_{s'}}{4} \prod_{s''\in S\setminus\{s,s'\}:\, p_{s,s'}\in s''} \left(1-\frac{y_{s''}}{2}\right)\\
&\ge& \frac{y_s y_{s'}}{4} \left( 1\ - \sum_{s''\in S\setminus\{s,s'\}:\, p_{s,s'}\in s''} \frac{y_{s''}}{2} \right)\ \ge\  
\frac{y_s y_{s'}}{8}.
\end{eqnarray*}
So,
\[ \Ex[|\Psi_t|]\ \ge\  \sum_{\begin{subarray}{c}s,s'\in S:\\ w_{s'}>w_s\ge t,\ s\cap s'\cap P\neq\emptyset\end{subarray}} \frac{y_sy_{s'}}{8}.
\]
It follows that
\[ \sum_{\begin{subarray}{c}s,s'\in S:\\ w_{s'}>w_s\ge t,\ s\cap s'\cap P\neq\emptyset\end{subarray}} y_sy_{s'}
\ \le\ 4c \sum_{s\in S:\, w_s\ge t} y_s.
\]
Integrating both sides over all $t$ from 0 to $\infty$, we get
\[ \sum_{\begin{subarray}{c}s,s'\in S:\\ w_{s'}>w_s,\ s\cap s'\cap P\neq\emptyset\end{subarray}} w_sy_sy_{s'}
\ \le\ 4c \sum_{s\in S} w_sy_s\ =\ 4c\,\OPTDISLP.
\]
We conclude that 
\[ \Ex\left[\sum_{s\in I}w_s\right]\ \ge\ \frac{\OPTDISLP}{b} - \frac{4c\,\OPTDISLP}{b^2}\ =\ \frac{\OPTDISLP}{16c}
\]
for $b=8c$.  (The algorithm can be derandomized by the standard method of conditional probabilities or expectations~\cite{MotwaniR95}.)
\qed

}

\begin{corollary}\label{cor:DIS}
Suppose $\SSS$ can be decomposed into $c'$ subclasses $\SSS_1,\ldots,\SSS_{c'}$,
each with 2-shallow cell complexity $f_2(n)\le cn$ for a constant $c$.
We can find a feasible solution to the max-weight DIS problem of weight $\Omega(\OPTDISLP/(cc'))$ in polynomial time.
\end{corollary}
\begin{proof}
Pick a subclass $\SSS_j$ with $\sum_{s\in S\cap\SSS_j}w_sy_s\ge \OPTDISLP/c'$,
and then apply Lemma~\ref{lem:DIS} to $S\cap\SSS_j$.
\end{proof}

\section{LP-Based Approximation Algorithm for Max Coverage}\label{sec:max:cov:lp}

In this section, we present our main approximation algorithm for the max coverage problem, defined as follows:

\begin{definition}
In the \emph{weighted (or budgeted) max coverage} problem, we are
given a set $S\subseteq\SSS$ of objects and a set $P\subseteq\UUU$ of points, 
where each object $s$ has a weight $w_s>0$,
and we want to find a subset of objects in $S$ covering the largest number of
points in $P$, such that the total weight of the subset is at most a given budget, which by rescaling we assume to be $1$.

The standard LP relaxation for weighted max coverage is the following:
\[\begin{array}{lll}
\text{max}  & \displaystyle\sum_{p\in P} z_p'\\
\text{s.t.} & z_p'\,\le\, z_p := \displaystyle\sum_{s\in S:\,p\in s} x_s  & \forall p\in P\\
            & 0\le z_p'\le 1                                         & \forall p\in P\\[1ex]
            & \displaystyle\sum_{s\in S} w_sx_s \le 1   & \\
            &  0\le x_s\le 1                   & \forall s\in S.
\end{array}\]
\end{definition}

Our (weighted) max coverage algorithm is actually a general black-box reduction to (weighted) DIS:

\begin{theorem}\label{thm:maxcov:LP}
Suppose that for any $S\subseteq\SSS$ and $P\subseteq\UUU$, we can find a feasible solution to the max-weight DIS problem of size at least $\OPTDISLP/C$ in polynomial time, where
$\OPTDISLP$ denotes the optimal LP value for max-weight DIS.

Then for any $S\subseteq\SSS$ and $P\subseteq\UUU$, we can find a feasible solution
to the weighted max coverage problem of value at least $(1-1/e+\Omega(1/C^2))\OPTLP-\vmax$ in polynomial time,
where $\OPTLP$ denotes the optimal LP value for weighted max coverage and $\vmax:=\max_{s\in S}|s\cap P|$ is the maximum number of points covered by a single object.  The term $\vmax$ can be removed in the unweighted case (i.e., when all the individual weights are equal to $1/k$).
\end{theorem}

\begin{proof}
Let $(x_s)_{s\in S}$, $(z_p)_{p\in P}$, $(z_p')_{p\in P}$ be the optimal LP solution.

\paragraph{Algorithm.}  Let $\delta$ be $1/C$ times  a sufficiently small constant, and let $\delta'=5\delta$.

\newcommand{\aaa}{\hat{a}}
\newcommand{\AAA}{\hat{A}}
\newcommand{\bbeta}{\hat{\beta}}
\newcommand{\kkk}{\hat{k}}
\newcommand{\ttt}{\hat{t}}

\begin{enumerate}
\item[0.] Let $\Pbad = \{p\in P: z_p > 2\}$ and $\Pgood = P\setminus\Pbad$.  

\item If $|\Pbad| > \delta\,\OPTLP$, then 
return a solution by {\bf LP rounding}.

\item Otherwise, do {\bf standard greedy}: For $i=0,1,\ldots$, 
pick $a_{i+1}\in S$ maximizing $|a_{i+1}\cap P\setminus A_i|/w_{a_{i+1}}$ where
$A_i=(a_1\cup\cdots\cup a_i)\cap P$. 
Let $k$ be the smallest index such that $w_{a_1}+\cdots+w_{a_k}\ge 1$. 

Let $h$ be the smallest index such that $|a_{h+1}\cap P\setminus A_h|/w_{a_{h+1}} \le (1+\delta) (\OPTLP-|A_h|)$.  
If $|A_h| > \delta\,\OPTLP$, then return $\{a_1,\ldots,a_{k-1}\}$ (or, in the unweighted case, return $\{a_1,\ldots,a_k\}$).

\item Otherwise, let $\Sbad=\{s\in S: |s\cap \Pgood\setminus A_h|/w_s < (1-\delta')(\OPTLP-|A_h|)\}$ and $\Sgood=S\setminus\Sbad$.

\item Compute a feasible solution $I$ to the {\bf max-weight DIS} problem for the objects in $\Sgood$ and the points in $\Pgood$.

\item Do {\bf modified greedy}: Let $\aaa_1,\ldots,\aaa_{|I|}$ be the objects of $I$.
For $i=|I|,|I|+1,\ldots$, 
pick $\aaa_{i+1}\in S$ maximizing $|\aaa_{i+1}\cap P\setminus\AAA_i|/w_{\aaa_{i+1}}$ where
$\AAA_i= (\aaa_1\cup\cdots\cup \aaa_i)\cap P$. 
Let $\kkk$ be the smallest index such that $w_{\aaa_1}+\cdots+w_{\aaa_{\kkk}}\ge 1$.
Return $\{\aaa_1,\ldots,\aaa_{\kkk-1}\}$ (or, in the unweighted case, return $\{\aaa_1,\ldots,\aaa_{\kkk}\}$).
\end{enumerate}

\paragraph{Intuition.}
To develop intuition on why this algorithm could beat $1-1/e$, the following paragraphs provide some
informal explanations (they may be skipped for readers who prefer seeing the complete formal arguments directly). It is helpful to concentrate on the unweighted case (when all $w_s$'s are equal to $1/k$).

The algorithm returns one of 3 possible solutions: a known LP-rounding solution,
the greedy solution, or the solution of the modified greedy algorithm which is pre-loaded with a DIS.
If $|\Pbad| > \delta\,\OPTLP$, then by inspecting a known analysis of LP rounding,
it is not difficult to see that we gain an advantage over $1-1/e$.
On the other hand, if $|A_h| > \delta\,\OPTLP$, then by inspecting the known analysis
of the greedy algorithm, it is also not difficult to see that we gain an advantage
over $1-1/e$.  So, it suffices to focus on the case when $\Pbad$ and $A_h$ are both small.
For simplicity, let's pretend that $\Pbad$ and $A_h$ are both empty.

By definition of $h$, we then have $|s\cap P|\le (1+\delta)\OPTLP/k$ for all objects $s\in S$.
We know that $\sum_{s\in S} |s\cap P|x_s \ge \OPTLP$ and $\sum_{s\in S}x_s=k$.
If $(x_s)_{s\in S}$ is integral, this would imply that most of the terms $|s\cap P|$ for the $k$ objects $s$ in the solution must be close to $\OPTLP/k$---in other words, most of the objects in the solution are in $\Sgood$.
If $(x_s)_{s\in S}$ is not integral, it still implies that $\Sgood$ has most of the LP mass (i.e., $\sum_{s\in \Sgood}x_s \ge \Omega(k)$).  We can then find a large DIS of $\Sgood$ of size $\Omega(k)$ (since $(x_s/2)_{s\in S}$ is a feasible fractional solution to DIS, assuming $\Pbad$ is empty).

Now, a DIS $I$ of $\Sgood$ provides \emph{very} good coverage, since the subsets $s\cap P$ over all $s\in I$ are completely disjoint and have roughly the same size, near $\OPTLP/k$.  This gives us an advantage over greedy for the first $|I|$ iterations.  Continuing the greedy algorithm till the $k$-th iteration preserves this advantage.


\paragraph{Analysis of LP rounding (line~1).}
We use a known rounding procedure for weighted max coverage by Doerr, K\"unnemann, and Wahlstr\"om~\cite[Theorem~1]{DoerrKW10} (based on \cite{AgeevS04} in the unweighted case).\footnote{
In the unweighted case, we can alternatively use a more standard randomized rounding scheme: namely, for each of $k$ iterations, independently pick an object under the distribution that $s$ is selected with probability $x_s/k$.
However, Doerr et al.'s scheme has the advantages that it works in the weighted case and is deterministic (without needing derandomization later).
Srinvasan~\cite{Srinivasan01} described another rounding scheme for the weighted case which is more complicated.
}
They defined  
\[ F((x_s)_{s\in S})\ :=\ \sum_{p\in P} \left( 1-\prod_{s\in S:\,p\in s}(1-x_s)\right).
\]
If $(x_s)_{s\in S}$ is integral,  then $F((x_s)_{s\in S})$ corresponds to precisely the number of points covered by the solution.
For the optimal LP solution $(x_s)_{s\in S}$,
\begin{eqnarray*}
F((x_s)_{s\in S})
&\ge & \sum_{p\in P} \left(1 - \prod_{s\in S:\, p\in s} e^{-x_s}\right)
\ =\ \sum_{p\in P} (1-e^{-z_p})\\
&\ge & \sum_{p\in\Pgood} (1-e^{-z_p'}) + (1-e^{-2})|\Pbad|\\
&\ge & \sum_{p\in\Pgood} (1-e^{-1})z_p' + (1-e^{-2})|\Pbad| \qquad\mbox{(since $1-e^{-z}\ge (1-e^{-1})z$ for $z\in [0,1]$)}\\
&\ge & (1-e^{-1})(\OPTLP-|\Pbad|) + (1-e^{-2})|\Pbad|
\\
&\ge& (1-e^{-1})\OPTLP + \Omega(|\Pbad|)
\ \ge\ (1-e^{-1}+\Omega(\delta))\OPTLP,
\end{eqnarray*}
if $|\Pbad| > \delta\,\OPTLP$. 
Doerr et al.\ obtained a vector $(\xxx_s)_{s\in S}$ which is integral except
for one entry, such that $\sum_{s\in S} w_s\xxx_s\le 1$, $0\le \xxx_s\le 1$, and $F((\xxx_s)_{s\in S})\ge F((x_s)_{s\in S})$.
Rounding the non-integral entry downward, we get a feasible solution covering
at least $(1-e^{-1}+\Omega(\delta))\OPTLP - \vmax$ points.
(In the unweighted case, we can instead round the non-integral entry upward and avoid the $\vmax$ term.)
From now on, we may assume $|\Pbad|\le \delta\,\OPTLP$.

\paragraph{Analysis of standard greedy (line~2).}
We adapt known analysis of the
standard greedy algorithm
(e.g., see \cite{KhullerMN99} for the weighted case).


Let $\beta_i = w_{a_1}+\cdots+w_{a_i}$.
Let $t_i = \OPTLP - |A_i|$.  
Let $\Delta_i = |a_{i+1}\cap P\setminus A_i|/w_{a_{i+1}}$.
Then $t_i - t_{i+1} = |A_{i+1}|-|A_i| = w_{a_{i+1}}\Delta_i = (\beta_{i+1}-\beta_i)\Delta_i$.

For every $s\in S$, we know $|s\cap P\setminus A_i|/w_s\le \Delta_i$.  So,
\begin{eqnarray*}
\Delta_i &\ge& \sum_{s\in S} \Delta_i w_sx_s\ \ge\ \sum_{s\in S} |s\cap P\setminus A_i|\,x_s\ =\ \sum_{s\in S} \sum_{p\in P\setminus A_i:\, p\in s} x_s\ =\ \sum_{p\in P\setminus A_i} \sum_{s\in S:\, p\in s}x_s\\
&\ge& \sum_{p\in P\setminus A_i} z_p'\ \ge\ \OPTLP - |A_i|\ =\ t_i.
\end{eqnarray*}
Thus, 
\begin{equation}\label{eqn1}
t_{i+1}\ \le\ (1-(\beta_{i+1}-\beta_i)) t_i\ \le\ e^{-(\beta_{i+1}-\beta_i)}t_i \qquad \forall i\in\{0,1,\ldots\}.
\end{equation}

On the other hand, by definition of $h$, we have $\Delta_h \le (1+\delta) t_h$ 
and $\Delta_i > (1+\delta)t_i$ for $i\in\{0,\ldots,h-1\}$.
Thus,
\begin{equation}\label{eqn2} t_{i+1}\ \le\ \left(1-(1+\delta)(\beta_{i+1}-\beta_i)\right) t_i\ \le\ e^{-(1+\delta)(\beta_{i+1}-\beta_i)}t_i\qquad \forall i\in\{0,\ldots,h-1\}.
\end{equation}

Suppose $|A_h| > \delta\,\OPTLP$, i.e., $t_h < (1-\delta)\OPTLP$.
Recall that $\beta_k\ge 1$ and $\beta_0=0$.
If $k\le h$, then (\ref{eqn2}) implies $t_k \le e^{-(1+\delta)}\OPTLP$.
Otherwise, if $\beta_h\le\delta/2$, then (\ref{eqn1}) implies $t_k\le e^{-(1-\beta_h)}t_h < e^{-(1-\beta_h)}(1-\delta)\OPTLP\le e^{-1+\beta_h-\delta}\OPTLP = e^{-1-\Omega(\delta)}\OPTLP$.
On the other hand, if $\beta_h >\delta/2$, then
(\ref{eqn1}) and (\ref{eqn2}) imply
$t_k\le e^{-(1-\beta_h)}t_h\le e^{-(1-\beta_h)}e^{-(1+\delta)\beta_h}\OPTLP = e^{-1-\delta\beta_h}\OPTLP \le e^{-1-\Omega(\delta^2)}\OPTLP$.  In any case, 
the number of points covered by
$\{a_1,\ldots,a_{k-1}\}$ is 
$|A_{k-1}|\ge |A_k|-\vmax\ge 
(1-e^{-1-\Omega(\delta^2)})\OPTLP-\vmax \ge 
(1-e^{-1}+\Omega(\delta^2))\OPTLP-\vmax$. 
(In the unweighted case, we return $\{a_1,\ldots,a_k\}$ instead and avoid the $\vmax$ term.)

From now on, we may assume $|A_h| \le \delta\,\OPTLP$, i.e., $t_h  \ge (1-\delta)\OPTLP$.

\paragraph{On $\Sbad$ (line~3).}
Let $\gamma = 1- \sum_{s\in S} w_sx_s \ge 0$.
For every $s\in S$, we know $|s\cap P\setminus A_h|/w_s \le \Delta_h < (1+\delta) t_h$.  Now, 
\begin{eqnarray*}
\sum_{s\in S} |s\cap \Pgood\setminus A_h|\,x_s
&=& \sum_{s\in S}\sum_{p\in \Pgood\setminus A_h:\, p\in s} x_s
\ \ =\ \sum_{p\in \Pgood\setminus A_h} \sum_{s\in S:\, p\in s} x_s \ \ge\ \sum_{p\in \Pgood\setminus A_h} z_p'\\
&\ge& \OPTLP - |A_h| - |\Pbad|\ \ge\ t_h - \delta\,\OPTLP\ >\ (1-2\delta)t_h,
\end{eqnarray*}
implying that
\[
\sum_{s\in S} ((1+\delta)t_h w_s - |s\cap \Pgood\setminus A_h|)\,x_s
\ \le\ (1+\delta)t_h(1-\gamma) - (1-2\delta)t_h\ <\  (3\delta-\gamma) t_h.
\]
All the terms in the above sum are nonnegative.
Recall that $\Sbad = \{s\in S: |s\cap \Pgood\setminus A_h| < (1-\delta')t_hw_s\}$.
It follows that
\[ \sum_{s\in \Sbad}w_sx_s\ <\ \frac{(3\delta-\gamma)t_h}{(1+\delta)t_h - (1-\delta')t_h}
\ =\  \frac{3\delta-\gamma}{6\delta}\ =\ \frac12 - \frac{\gamma}{6\delta} \qquad\mbox{for $\delta'=5\delta$.}
\]

\paragraph{Analysis of DIS (line~4).}
Let $y_s=x_s/2$.
Then $(y_s)_{s\in \Sgood}$
is a feasible solution to the max-weight DIS LP for $\Sgood$ and $\Pgood$
(since $\sum_{s\in\Sgood:\, p\in s}y_s \le z_p/2\le 1$
for $p\in\Pgood$) and has weight
\[ \sum_{s\in \Sgood} w_sx_s/2 \ =\ \frac12\left(1-\gamma-\sum_{s\in\Sbad} w_sx_s\right)\ 
\ >\ \frac12 \left(1-\gamma - \frac{1}{2} + \frac{\gamma}{6\delta}\right) \ \ge\ \frac14.
\]

Thus, we can compute a feasible solution $I$  to max-weight DIS for $\Sgood$ and $\Pgood$
with $\sum_{s\in I} w_s \ge 1/(4C)$.

\paragraph{Analysis of modified greedy (line~5).}
Let $\bbeta_i = w_{\aaa_1}+\cdots+w_{\aaa_i}$.
Let $\ttt_i = \OPTLP - |\AAA_i|$.  
As before, 
\begin{equation}\label{eqn3} \ttt_{i+1}\ \le\ (1-(\bbeta_{i+1}-\beta_i)) \ttt_i\ \le\ e^{-(\bbeta_{i+1}-\bbeta_i)}\ttt_i \qquad \forall i\in\{|I|,|I|+1,\ldots\}.
\end{equation}

Recall that $\Sgood = \{s\in S: |s\cap \Pgood\setminus A_h| \ge (1-\delta')t_hw_s\}$.
For $i\in\{1,\ldots,|I|\}$, since no point of $\Pgood$ lies in more than one object of $I$,
\[ |\AAA_i|\ \ge\
\sum_{j=1}^i |\aaa_j\cap \Pgood \setminus A_h| 
\ \ge\ (1-\delta') t_h \sum_{j=1}^i w_{\aaa_j}
\ =\ (1-\delta') \bbeta_i t_h\ \ge\ (1-O(\delta))\bbeta_i\OPTLP.
\]
Thus, 
\begin{equation}\label{eqn4} \ttt_i\ \le\ (1-(1-O(\delta))\bbeta_i)\OPTLP\ \le\ e^{-(1-O(\delta))\bbeta_i - \Omega(\bbeta_i^2)}\OPTLP \qquad\forall i\in\{1,\ldots,|I|\}.
\end{equation}

Recall that $\bbeta_{\kkk} \ge 1$, and $\bbeta_{|I|}=\sum_{s\in I}w_s \ge 1/(4C)$.
If $\kkk \le |I|$, then (\ref{eqn4}) implies $\ttt_{\kkk}\le e^{-1+O(\delta)-\Omega(1)}\OPTLP \le e^{-1-\Omega(1)}\OPTLP$.
Otherwise, (\ref{eqn3}) and (\ref{eqn4}) imply
\begin{eqnarray*}
 \ttt_{\kkk} &\le& e^{-(1-\bbeta_{|I|})}\ttt_{|I|}
\ \le\  e^{-(1-\bbeta_{|I|})}e^{-(1-O(\delta))\bbeta_{|I|} - \Omega(\bbeta_{|I|}^2)}\OPTLP\ =\ e^{-1+O(\delta\bbeta_{|I|}) - \Omega(\bbeta_{|I|}^2)}\OPTLP\\
&\le& e^{-1-\Omega(\bbeta_{|I|}^2)}\OPTLP\ \le\ e^{-1-\Omega(1/C^2)}\OPTLP,
\end{eqnarray*}
by choosing $\delta$ to be $1/C$ times a sufficiently small constant.
So, the number of points covered by $\{\aaa_1,\ldots,\aaa_{\kkk-1}\}$ is at least
$(1-e^{-1-\Omega(1/C^2)})\OPTLP-\vmax\ge
(1-e^{-1}+\Omega(1/C^2))\OPTLP - \vmax$.
(In the unweighted case, we return $\{\aaa_1,\ldots,\aaa_{\kkk}\}$ instead and avoid the $\vmax$ term.)
\end{proof}

The following observation, essentially shown 
by Khuller, Moss, and Naor~\cite[Section~3]{KhullerMN99},
removes the $\vmax$ term for the weighted max coverage problem (though the approximation is no longer LP-relative%
\footnote{This is unavoidable, since the problem generalizes the 0-1 knapsack problem, which has integrality gap~2 for its standard LP.}%
).

\begin{lemma}\label{lem:vmax}
Suppose that for any $S\subseteq\SSS$ and $P\subseteq\UUU$, we can find a feasible solution to the 
weighted max coverage problem of value at least $\gamma\,\OPT-\vmax$ in polynomial time for a constant $\gamma<1$,
where $\OPT$ denotes the optimal value  and $\vmax:=\max_{s\in S}|s\cap P|$.

Then for any $S\subseteq\SSS$ and $P\subseteq\UUU$, we can find a feasible solution to the 
weighted max coverage problem of value at least $\gamma\,\OPT$ in polynomial time.
\end{lemma}
\begin{proof}
Let $b$ be a sufficiently large constant. For $i=0,\ldots,b$, we guess the object $a_{i+1}^*$ in the optimal solution that maximizes $|a_{i+1}^*\cap P\setminus A_i^*|$ where $A_i^*=(a_1^*\cup\cdots\cup a_i^*)\cap P$.
We then solve the problem for $P\setminus A_b^*$ and $\{s\in S: |s\cap P\setminus A_b^*|\le |a_{b+1}^*\cap P\setminus A_b^*|\}$ with budget $1-(w(a_1^*)+\cdots+w(a_b^*))$, and add back $a_1^*,\ldots,a_b^*$
to the solution.  We return the best solution over all $O(n^{b+1})$ guesses.  (If the optimal solution has fewer than $b+1$ objects,
we can find it directly by trying all $O(n^b)$ solutions.)

Let  $\vmax := |a_{b+1}^*\cap P\setminus A_b^*|\le |a_{b+1}^*\cap P\setminus A_i^*|\le |a_{i+1}^*\cap P\setminus A_i^*| =|A_{i+1}^*|-|A_i^*|$
for each $i\in\{0,\ldots,b-1\}$.  Summing yields $b\vmax\le |A_b^*|$, i.e., $\vmax\le |A_b^*|/b$.
The final solution for the correct guess has value at least
$|A_b^*| + \gamma (\OPT-|A_b^*|) - |A_b^*|/b = \gamma\,\OPT + (1-\gamma-1/b)|A_b^*| \ge \gamma\,\OPT$
by choosing $b\ge 1/(1-\gamma)$.
\end{proof}

Putting Theorem~\ref{thm:maxcov:LP} and Lemma~\ref{lem:vmax}
together with Lemma~\ref{lem:DIS} and Corollary~\ref{cor:DIS}, we now obtain the following immediate consequences:

\begin{corollary}
For the following families of objects,
there are polynomial-time $(1-1/e+\Omega(1))$-approximation algorithms for
the weighted max coverage problem:
\begin{enumerate}
\item[(i)] 2D disks/pseudodisks/$r$-admissible objects, 2D fat rectangles, 2D similar-height rectangles,
2D rectangles stabbable by one horizontal line,
2D similar-size fat rectangles, 2D homothets of one convex object, 3D halfspaces, 3D unit cubes, 3D similar-size fat boxes, and 3D boxes stabbable by one point;
\item[(ii)] 2D fat rotated rectangles with $O(1)$ orientations, 2D similar-height rotated rectangles with $O(1)$ orientations, 2D homothets of $O(1)$ convex objects, 3D boxes stabbable by $O(1)$ points, $O(1)$-dimensional slabs with $O(1)$ orientations.
\end{enumerate}
\end{corollary}
\begin{proof}
The families in (i) have linear union complexity, and thus linear 2-shallow cell complexity, by known facts from combinatorial geometry~\cite{AgarwalPS08} (see \cite[Theorem~1.3]{MatousekPSSW94} on similar-size fat triangles).
The families in (ii) can be decomposed into $O(1)$ sub-families that have linear union complexity and 2-shallow complexity.
\end{proof}

We can also consider the \emph{weighted max hitting set} problem: select a subset of points of $P$ with total weight at most a given budget, maximizing the number of objects of $S$ that are hit by the subset.
Max hitting set is equivalent to max coverage in the dual set system,
where the roles of points and objects are reversed.
By known combinatorial bounds on shallow cell complexity in the dual system (see~\cite{BuzagloPR13} on pseudodisks),
we obtain:

\begin{corollary}
For the following families of objects,
there are polynomial-time $(1-1/e+\Omega(1))$-approximation algorithms for
the weighted max hitting set problem:
\begin{itemize}
\item 2D disks/pseudodisks,
2D homothets of $O(1)$ convex objects,
3D halfspaces, 3D unit cubes, and
$O(1)$-dimensional slabs with $O(1)$ orientations.
\end{itemize}
\end{corollary}

We can also consider the \emph{weighted max dominating set} problem: 
select a subset of objects of $S$ with total weight at most a given budget, maximizing the number of objects of $S$ that are intersected by the subset.
Max dominating set reduces to max coverage in an appropriate set system.
By known combinatorial bounds on shallow cell complexity in such set systems \cite{AronovDEP21},
we obtain:

\begin{corollary}
For the following families of objects,
there are polynomial-time $(1-1/e+\Omega(1))$-approximation algorithms for
the weighted max dominating set problem:
\begin{itemize}
\item 2D disks/pseudodisks and
2D homothets of $O(1)$ convex objects.
\end{itemize}
\end{corollary}

For an application to planar graphs, Browne and Chang~\cite{BrowneC26} recently studied
the problem of finding minimum-weight subset of vertices whose distance-at-most-$r$ neighborhoods
cover all vertices in an unweighted planar graph.  They obtained an $O(1)$-approximation algorithm using
the set cover framework in \cite{ChanGKS12} by proving bounded shallow-cell complexity for the associated set system.
We immediately obtain a $(1-1/e+\Omega(1))$-approximation algorithm for the problem 
of finding a subset of vertices with total weight at most a budget, maximizing the number of vertices covered by their distance-at-most-$r$ neighborhoods.

For a non-algorithmic application of Theorem~\ref{thm:maxcov:LP}, we have the following new, interesting combinatorial result:

\begin{corollary}
For any set $P$ of $m$ points and any set $S$ of $n$ halfspaces in $\R^3$, such that each point is
contained in at least $n/k$ halfspaces of $S$, there exists $k$ halfspaces of $S$
which covers at least $(1-1/e+\Omega(1))m$ points of $P$. 
(The same holds for pseudo-disks/$r$-admissible objects in $\R^2$.)
\end{corollary}
\begin{proof}
Here, $\OPTLP=m$ for the (unweighted) max coverage LP,
by setting $x_s=k/n$ for all $s\in S$ and $z_p'=1$ for all $p\in P$.
\end{proof}

\section{Parameterized Approximation Scheme for Max Coverage}\label{sec:improve:BKL}

In this section, we consider the (unweighted) max coverage problem in the setting when $k$ is small.
For any set system with VC dimension $d$, Badanidiyuru, Kleinberg, and Lee~\cite{BadanidiyuruKL12} presented an FPT approximation scheme, specifically, a $(1-\eps)$-approximation algorithm running in $2^{O((dk^2/\eps^5)\polylog(dk/\eps))}\poly(|P|,|S|)$ time.
We simplify their algorithm and improve their time bound to $2^{O((dk/\eps)\log(dk/\eps))}\poly(|P|,|S|)$.

Badanidiyuru, Kleinberg, and Lee's key observation is that either the greedy solution yields a good enough approximation, or
the portion $A^G$ of the universe covered by the greedy solution have substantial overlap with the portion $A^*$ of the universe covered by the optimal solution.  We cannot afford to guess $A^*\cap A^G$ exactly.
Instead, we approximate $A^G$ by a small subset $R$ of points (good approximations exist for set systems with bounded VC dimension), guess $Q=A^*\cap R$, and guess how each object $a_i^*$ in the optimal solution $A^*$ covers $R$.
Then we can remove $A^G$ and recurse.

We first recall known facts about set systems of bounded VC dimension:

\begin{definition}
For a set system $(\UUU,\SSS)$ and a set $P\subseteq\UUU$,
a subset $R\subseteq P$ is called a \emph{$\delta$-approximation of $P$} if for every $\gamma\in\SSS$,
\[ \left| \frac{|\gamma\cap R|}{|R|} - \frac{|\gamma\cap P|}{|P|} \right|
\ \le\ \delta.
\]
\end{definition}

\begin{fact}[\cite{BronnimannCM99}] \label{fact:approx-net}
Let $(\UUU,\SSS)$ be a set system with VC dimension $D$ equipped with a subsystem oracle.
Given a set $P\subseteq\UUU$ of size $n$ and $\delta>0$,
there exists a $\delta$-approximation of $P$ of size
$O(\frac{D}{\delta^2}\log\frac{D}{\delta})$, and it can be constructed
in $O(D)^{3D} (\frac{1}{\delta^2}\log\frac{D}{\delta})^D n$ time.
\end{fact}

\begin{fact}[observed by \cite{BadanidiyuruKL12}]
Let $(\UUU,\SSS)$ be a set system with VC dimension $d$.
Define $\SSS^{\cup k}$ to consist of all sets of the form $s_1\cup\cdots\cup s_k$
where $s_1,\ldots,s_k\in\SSS$.
Then $(\UUU,\SSS^{\cup k})$ has VC dimension at most $kd$.
\end{fact}

\paragraph{Algorithm.}
Let $(\UUU,\SSS)$ be a set system with VC dimension $d$.
We will solve a generalization of the max $k$-coverage problem:
given a set $P\subseteq\UUU$ and sets $S_1,\ldots,S_k\subseteq\SSS$ of total size $n$, find
a $k$-tuple $(a_1,\ldots,a_k)\in S_1\times\cdots\times S_k$ maximizing
the number of points covered, i.e., $|(a_1\cup\cdots\cup a_k)\cap P|$.
(The original problem corresponds to the case when $S_1=\cdots=S_k=S$.)

Let $\delta_\ell,\eps_\ell>0$ be parameters to be set later. 
The following algorithm outputs a solution to this problem with approximation factor
at least $1-\eps_\ell$.

\begin{quote}
\begin{tabbing}
10.\ \ \= for\ \= for\ \= for\ \=\kill
$\MAXCOV_\ell(S_1,\ldots,S_k,P):$\\[1ex]
0.\> if $\ell=0$, return an arbitrary $k$-tuple in $S_1\times\cdots\times S_k$\\[1ex]
1.\> for $i=0,1,\ldots,k-1$:\ \ \ \ {\bf (first compute greedy solution)}\\
2.\>\> pick $a^G_{i+1}\in S_{i+1}$ maximizing $|s\cap P\setminus A^G_i|$ where
$A^G_i:=(a^G_1\cup\cdots\cup a^G_i)\cap P$\\
3.\> add $(a^G_1,\ldots,a^G_k)$ to a list $L$ and let $A^G=A^G_k$\\[1ex]
4.\> let $R$ be a $\delta_\ell$-approximation of $A^G$ in the set system
$(\UUU,\SSS^{\cup k})$\\
5.\> let $\SSS_{\mid R} = \{s\cap R: s\in \SSS\}$\\[1ex]
6.\> for each $\sigma_1,\ldots,\sigma_k\in \SSS_{\mid R}$:\ \ \ \ {\bf (then guess and recurse)}\\
7.\>\> for $i=1,\ldots,k$, let $S^{(\sigma_i)}_i = \{s\in S_i: s\cap R=\sigma_i\}$\\
8.\>\> add $\MAXCOV_{\ell-1}(S^{(\sigma_1)}_1,\ldots,S^{(\sigma_k)}_k,P\setminus A^G)$ to $L$\\[1ex]
9.\> return the $k$-tuple in $L$ that covers the most points of $P$ 
\end{tabbing}
\end{quote}

\paragraph{Analysis of approximation factor.}
Let $(a_1^*,\ldots,a_k^*)$ denote the optimal $k$-tuple in $S_1\times\cdots\times S_k$ 
maximizing $|A^*|$ where $A^*=(a_1^*\cup\cdots\cup a_k^*)\cap P$.  

If $|A^G|\ge (1-\eps_\ell)|A^*|$, then the greedy solution $(a^G_1,\ldots,a^G_k)$ already has approximation factor
at least $1-\eps_\ell$.  So, we may assume $|A^G| < (1-\eps_\ell)|A^*|$ from now on.

BKL observes that $|A^*\cap A^G| > \eps_\ell |A^*|$, which can be shown as follows:
For each $i\in\{0,\ldots,k-1\}$, $|a_{i+1}^*\cap P\setminus A^G|\le |a_{i+1}^*\cap P\setminus A^G_i| \le |a^G_{i+1}\cap P\setminus A^G_i|=|A^G_{i+1}|-|A^G_i|$.
Summing over all $i$ gives $|A^*\setminus A^G|\le |A^G|$,
and so $|A^*\cap A^G| \ge |A^*|-|A^G| > \eps_\ell |A^*|$.

Let $\sigma_i=a_i^*\cap R$. 
Let $(a_1,\ldots,a_k)$ be the $k$-tuple returned by 
$\MAXCOV_{\ell-1}(S^{(\sigma_i}_1,\ldots,S^{(\sigma_k)}_k,P\setminus A^G)$.
By definition of $\delta_\ell$-approximation,
\[ \left| \frac{|A\cap R|}{|R|} - \frac{|A\cap A^G|}{|A^G|} \right|
\ \le\ \delta_\ell 
\qquad
\mbox{and}
\qquad
\left| \frac{|A^*\cap R|}{|R|} - \frac{|A^*\cap A^G|}{|A^G|} \right|
\ \le\ \delta_\ell.
\]
Since $a_i\in S^{(\sigma_i)}_i$, we have $a_i\cap R = \sigma_i = a_i^*\cap R$ for all $i\in\{1,\ldots,k\}$.
Thus, $A\cap R = A^*\cap R$, and so
\[ \left| |A\cap A^G| -  |A^*\cap A^G| \right|\ \le\ 2\delta_\ell |A^G|.
\]
We conclude that
\begin{eqnarray*}
 |A|\ =\ |A\setminus A^G| + |A\cap A^G|
     &\ge& (1-\eps_{\ell-1})|A^*\setminus A^G| + |A^*\cap A^G| - 2\delta_\ell |A^G|\\
     &\ge& (1-\eps_{\ell-1})(|A^*|-|A^*\cap A^G|) + |A^*\cap A^G| - 2\delta_\ell |A^G|\\
     &=& (1-\eps_{\ell-1})|A^*| + \eps_{\ell-1}|A^*\cap A^G| - 2\delta_\ell |A^G|\\
     &>&   (1-\eps_{\ell-1})|A^*| + \eps_{\ell-1}\eps_\ell |A^*| - 2\delta_\ell |A^*|\\
     &=&   (1-\eps_{\ell-1} + \eps_{\ell-1}\eps_\ell/2)|A^*| \qquad\mbox{by setting $\delta_\ell=\eps_{\ell-1}\eps_\ell/4$}\\ 
     &=& (1-\eps_\ell)|A^*|,
\end{eqnarray*}
assuming that
\[ \eps_\ell = \eps_{\ell-1}(1-\eps_\ell/2),
\]
with $\eps_0=1$.
This equation is satisfied by setting $\eps_\ell=\frac{2}{\ell+2}$ (and hence, $\delta_\ell=\eps_{\ell-1}\eps_\ell/4 = \frac{1}{(\ell+1)(\ell+2)}$).
Choosing $\ell=\lceil 2/\eps\rceil-2$ yields a $(1-\eps)$-approximation at the end.

\paragraph{Analysis of running time.}
In line~4, by the two facts, we can compute a $\delta_\ell$-approximation $R$ of size
$|R|=O(\frac{dk}{\delta_\ell^2}\log\frac{dk}{\delta_\ell})$ 
in $O(dk)^{3dk} (\frac{1}{\delta_\ell^2}\log\frac{dk}{\delta_\ell})^{dk}n$ time.
In line~5,  $|\SSS_{\mid R}|=O(|R|^d)$ by the Sauer--Shelah lemma~\cite{Har11}.
In line~6, the number of choices of $(\sigma_1,\ldots,\sigma_k)$
is $O(|\SSS_{\mid R}|^k)=O(|R|^{dk}) = O(2^{O(dk\log(dk/\delta_\ell))})$.
Recall that $\delta_\ell = \Theta(1/\ell^2)$.
Thus, the running time $T_\ell$ satisfies the recurrence
\[ T_\ell\ \le\ 2^{O(dk\log(dk\ell))}\cdot (T_{\ell-1} + n^{O(1)})
\ \ \Longrightarrow\ \ T_\ell\le 2^{O(dk\ell\log(dk\ell))}n^{O(1)}.
\]

\begin{theorem}
For set systems with VC dimension $d$, there is a $(1-\eps)$-approximation algorithm for the max coverage problem with running time $2^{O((dk/\eps)\log(dk/\eps))}\poly(|P|,|S|)$.
\end{theorem}


\paragraph{Remarks.}
There are a few differences with 
Badanidiyuru, Kleinberg, and Lee's original algorithm~\cite{BadanidiyuruKL12}:
\begin{itemize}
\item They bounded the number of choices more naively, without
considering $\SSS_{\mid R}$ and applying the Sauer--Shelah lemma,
resulting in a running time that is exponential in $k^2\polylog k$ instead of $k\polylog k$.
\item They described their algorithm iteratively.  Our description using recursion is a bit cleaner.
\item They used the greedy-and-guess approach to compute a subset $V$ of the universe $P$ whose solution approximates the optimal solution and whose cardinality is within a constant factor of the optimal value; afterwards, given this ``support set'' $V$, a separate lemma was applied to solve the original problem (this time using the Sauer--Shelah lemma).  We use the greedy-and-guess approach to solve the original problem \emph{directly}, thereby further simplifying the algorithm.
\end{itemize}

\title{A Shifted-Quadtree PTAS for Maximum Coverage by Fat Objects}
\author{}
\date{}
\maketitle
\end{comment}
\section{PTAS for Max-Volume Selection for Fat Objects}\label{sec:ptas}

In this section, we present an efficient PTAS for the continuous version of max coverage, i.e., max-volume selection, for fat convex objects in a constant dimension $d$.
Our algorithm works also for the weighted version of the problem.  It is based on the shifted quadtree technique~\cite{Chan03}.

\newcommand{\diam}{\text{diam}}
\newcommand{\ext}{\text{extension}}

We use the following definition of fat objects (which is clearly satisfied by balls and hypercubes):

\begin{definition}
For a constant $\alpha>0$,
an object $s$ is \emph{$\alpha$-fat} if it contains a ball of radius $\alpha\,\diam(s)$, where
$\diam(s)$ denotes the $L_\infty$-diameter of $s$.
\end{definition}

The following fact about fat objects will be needed.  We are not able to find an explicit reference, and so include a proof in Appendix~\ref{app:expansion}.

\begin{fact}\label{fact:expansion}
Let $B_r$ denote the ball of radius $r$ centered at the origin, and let $\oplus$ denote the Minkowski sum.
For any set $S$ of $\alpha$-fat convex objects in $\R^d$,
\[ \vol{\bigcup_{s\in S} (s\oplus B_{\eps\,\diam(s)})} \le (1+O_{d,\alpha}(\eps)) \vol{\bigcup_{s\in S} s}.\]
\end{fact}

\newcommand{\sss}{\widetilde{s}}

We are given a (weighted) set $S$ of $n$ $\alpha$-fat convex objects in $\R^d$, where $d$ and $\alpha$ are fixed constants.
We may assume that the maximum individual weight is at most the budget.

Let $\vmax=\max_{s\in S}\vol{s}$.
Since the optimal value is at least $\vmax$,
we may remove all objects $s\in S$ with $\vol{s}\le \eps\vmax/n$ (since these contribute at most $\eps\vmax$ total volume).  
So, we may assume that $\diam(s) \in [\Omega((\eps\vmax/n)^{1/d}),O(\vmax^{1/d})]$ for all $s\in S$.
If $S$ can be separated by a hyperplane orthogonal to the $x$-axis into two nonempty subsets, we can translate one of the subsets to bring them closer without affecting the optimal value.
Consequently, we may assume that the $x$-projection of the objects lies in an interval of length $O(n\vmax^{1/d})$.
Applying the same argument, we may assume that all objects lies in
$[0,O(n\vmax^{1/d})]^d$.
By rescaling to make $\vmax=\Theta(1/n^d)$,
we may ensure that all objects $s\in S$ lie in $[0,1)^d$ and have 
diameter $\Omega((\eps/n^{d+1})^{1/d})$.


A \emph{quadtree cell} refers to a hypercube of the form $[i_1/2^\ell,(i_1+1)/2^\ell)
\times\cdots\times [i_d/2^\ell,(i_d+1)/2^\ell)$ for some 
$i_1,\ldots,i_d,\ell\in\mathbb{Z}$.  The \emph{level} of the cell refers to $\ell$.

\newcommand{\lmax}{\lambda_{\text{max}}}

Let $b$ be a parameter to be set later.
For each object $s$, let $\lambda_s$ be the level of the smallest quadtree cell containing $s$.  Let  $\lmax = \max_{s\in S}\lambda_s \le \log((n^{d+1}/\eps)^{1/d}) + O(1)$.
We define the \emph{rounded object} $\sss$ associated with $s$ to be the union of
all level-$(\lambda_s+b)$ quadtree cells intersecting $s$.

We first use dynamic programming (DP) to solve the problem \emph{exactly} for the rounded objects (to deal with the budget constraint for the weighted problem, we incorporate a standard idea from a textbook approximation algorithm for knapsack~\cite{VazBOOK}):

\begin{lemma}\label{lem:dp}
We can solve the weighted max-volume selection problem \emph{exactly} for the
rounded objects $\{\sss: s\in S\}$ in $2^{O(2^{db})}\cdot 2^{O(\lmax)}\cdot \poly(|S|)$ time.
\end{lemma}
\begin{proof}
For a quadtree cell, let $S_\gamma$ denote the subset of all objects $s\in S$ contained in $\gamma$.

Given a level-$\ell$ quadtree cell $\gamma$, a number $v$, and 
a union $Z$ of level-$(\ell+b)$ quadtree cells contained in~$\gamma$,
define $F[\gamma,v,Z]$ to be the minimum weight of $Q$
over all subsets $Q\subseteq S_\gamma$ such that $\vol{\bigcup_{s\in Q} \sss \cup Z}\ge v$.
Note that $Z$ is a union of at most $2^{db}$ cells.
It suffices to consider values $v\in (0,O(n)]$ that are integer multiples of $2^{-(\lmax+b)d}$.
The optimal value is the largest $v$ for which $F[[0,1)^d,v,\emptyset]$ is at most the given budget.

To compute $F[\gamma,v,Z]$, divide $\gamma$ into $2^d$ level-$(\ell+1)$ quadtree cells $\gamma_1,\ldots,\gamma_{2^d}$.
Let $s_1,\ldots,s_h$ be a list of all objects in $S_\gamma\setminus \bigcup_{i=1}^{2^d} S_{\gamma_i}$;
these objects all have level $\ell$ and so $\sss_i$ is a union of level-$(\ell+b)$ quadtree cells.
Define $F_j[\gamma,v,Z]$ to be the minimum weight of $Q$ 
over all subsets $Q\subseteq S_\gamma\setminus \{s_{j+1},\ldots,s_h\}$ such that 
$\vol{\bigcup_{s\in Q} \sss \cup Z}\ge v$.
Then $F[\gamma,v,Z]=F_h[\gamma,v,Z]$, and for $j\in\{1,\ldots,h\}$,
\begin{equation}\label{rec1}
F_j[\gamma,v,Z] = \min\{F_{j-1}[\gamma,v,Z],\, F_{j-1}[\gamma,v,Z\cup\sss_j]+ w_{s_j}\},
\end{equation}
and 
\begin{equation}\label{rec2}
F_0[\gamma,v,Z] = \min_{v_1,\ldots,v_{2^d}:\, v_1+\cdots+v_{2^d}=v}
\sum_{i=1}^{2^d} F[\gamma_i,v_i,Z\cap\gamma_i].
\end{equation}
Note that $Z\cap\gamma_i$ is a union of level-$(\ell+b)$ quadtree cells, but can be converted
to a union of level-$(\ell+b+1)$ quadtree cells.

Using these recursive formulas, 
we can compute $F_j[\gamma,v,Z]$ for all $\gamma,v,Z,j$ by dynamic programming (with the obvious base cases).
The number of choices for nonempty grid cells $\gamma$ is $O(|S|\lmax)$.
The number of choices for $Z$ is at most $2^{2^{db}}$.
The number of choices for $v$ is $N=O(n2^{(\lmax+b)d})$.

(Note that (\ref{rec2}) naively requires $O(N^{2^d})$ time per $\gamma$ and $Z$, but can be evaluated more efficiently
by dynamic programming or by viewing the formula as a min-plus convolution of $2^d$ vectors.  With approximation,
the running time can be lowered further by borrowing known techniques for the knapsack problem, e.g., \cite{Chan18a}.)
\end{proof}

Given a parameter $L$, we say that an object $s$ is \emph{$L$-aligned} if it is contained in a quadtree cell of side length at most 
$L\cdot\diam(s)$.
We now use the above lemma to design an approximation algorithm when all objects are $L$-aligned.

\begin{lemma}\label{lem:solve:aligned}
If all objects in $S$ are $\alpha$-fat and $L$-aligned, then we can solve the weighted max-volume selection problem 
with approximation factor $1-O(\eps)$ in $2^{O((L/\eps)^d)}\poly(|S|)$ time.
\end{lemma}
\begin{proof}
Choose $b=\lceil\log(L/\eps)\rceil$.  
Let $\{\sss:s\in Q\}$ be the solution produced by Lemma~\ref{lem:dp}.  Let $Q^*$ be the optimal solution.

If $s$ is $L$-aligned, then $2^{-\lambda_s}\le L\cdot\diam(s)$ and $2^{-(\lambda_s+b)}\le  \eps\,\diam(s)$, 
implying that 
$\sss\subseteq s \oplus B_{O(\eps\,\diam(s))}$.  
Then
$\vol{\bigcup_{s\in Q^*}s}\le \vol{\bigcup_{s\in Q^*}\sss}\le \vol{\bigcup_{s\in Q}\sss}\le (1+O(\eps))\vol{\bigcup_{s\in Q}s}$
by Fact~\ref{fact:expansion}.
\end{proof}

Finally, we apply known shifting techniques to reduce to the aligned case.
Fix an odd integer parameter $L>d$. For every
$j\in\{0,\ldots,L-1\}$,
define the shift vector
$v^{(j)}=(j/L,\ldots,j/L)\in\R^d$.
Chan~\cite{Chan03} proved the following shifting lemma:

\begin{lemma}[Shifting lemma~\cite{Chan03}]
\label{lem:alignment}
Let $s$ be an object 
contained in $[0,1)^d$. 
Then $s+v^{(j)}$ is $2L$-aligned for all but at most $d$ indices
$j\in\{0,\ldots,L-1\}$.
\end{lemma}

\begin{theorem}
For fat convex objects in a constant dimension $d$,
there is a $(1-O(\eps))$-approximation algorithm for weighted max-volume selection
with running time $2^{O(1/\eps^{2d})}\poly(|S|)$.
\end{theorem}
\begin{proof}
Choose $L\approx d/\eps$.
Pick a random shift $v^{(j)}$, and solve the problem for the subset
$S^{(j)} := \{s\in S: \mbox{$s+v^{(j)}$ is $2L$-aligned}\}$ by Lemma~\ref{lem:solve:aligned}.

Let $Q^*$ be the optimal solution.
For a fixed point $p$ covered by $Q^*$, the probability that $p$ is covered by
a $2L$-aligned object is at least $1-d/L$ by the shifting lemma.
Then  $Q^*\cap S^{(j)}$ is a feasible solution for $S^{(j)}$ with volume at least 
$(1-d/L)\OPT$.  Thus, the returned solution has volume at least $(1-d/L)(1-O(\eps))\OPT \ge (1-O(\eps))\OPT$.

(The algorithm can be easily derandomized by trying all $L$ shifts and returning the best solution.)
\end{proof}

\newcommand{\cS}{\mathcal{S}}
\newcommand{\cU}{\mathcal{P}} 
\newcommand{\alg}{\mathrm{ALG}}

\section{Hardness Results}

In this section, we prove hardness of approximation results for max-volume selection and max coverage.

\subsection{Max-Volume Selection}

First, we start with the continuous (unweighted) max coverage problem, i.e., max-volume selection. We will prove hardness of approximation in the simplest case where the objects are axis-aligned boxes.

\subsubsection{Hardness in High Dimension}\label{sec:hard:vol:high}

In high-dimension (where $d$ is part of the input), we can rule out PTAS for the problem. 

\begin{theorem}
When the dimension $d$ is part of the input, the max-volume selection problem is NP-hard to approximate to within $(1 - \eps)$-factor for some constant $\eps > 0$, even when each object is a similar-size box.
\end{theorem}

\begin{proof}
We will reduce from the max independent set (Max-IS) problem on 3-regular graph which is known to be NP-hard to approximate to $(1 - \delta)$ factor for some constant $\delta > 0$~\cite{ChlebikC03}. Let $G = (V, E)$ and $k$ denote the input to Max-IS\@. We construct an input to the max-volume selection problem as follows:
\begin{itemize}
\item Let $d = |E|$ where each coordinate is associated with an edge. Pick two bounded intervals $I_0 \subseteq I_1 \subseteq \R$ of lengths $\ell_0 \leq \ell_1$ (e.g. $I_0 = [0, 1], I_1 = [0, 2]$ works for this reduction).
\item For every vertex $v$, create an object $O_v$ which is a box defined as $O_v = \prod_{e \in E} I_{\bone[v \in e]}$.
\item Finally, let $k$ remains the same.
\end{itemize}
It is obvious that the reduction runs in polynomial time. We will now prove its completeness and soundness. For this, it is convenient to define $B := \prod_{e \in E} I_0$.

\paragraph{(Completeness)} Suppose that there are $k$ vertices $v_1, \dots, v_k$ that form an independent set. Consider $O_{v_1} \cup \cdots \cup O_{v_k}$. It is simple to see that $O_{v_i} \setminus B$ are all disjoint (and that $O_{v_i} \supset B$). Thus, the volume is exactly
\begin{align*}
V := \ell_0^d + (\ell_1^3 \ell_0^{d - 3} - \ell_0^d)k \geq k \cdot \ell_0^d ((\ell_1 / \ell_0)^3 - 1).
\end{align*}

\paragraph{(Soundness)} Suppose that the maximum independent set has size less than $(1 - \delta)k$. Now, consider any $k$ objects $O_{v_1}, \dots, O_{v_k}$. Let $S$ denote any maximal independent set among $v_1, \dots, v_k$; by our assumption $|S| < (1 - \delta)k$. Note that every $v_i \notin S$ must have an edge to a vertex in $S$, let this edge be $e_i$ and the vertex be $u_i$.

Now, we have
\begin{align*}
&\vol{O_{v_1} \cup \cdots \cup O_{v_k}} \\
&= \vol{B} +  \vol{(O_{v_1} \setminus B) \cup \cdots \cup (O_{v_k} \setminus B)} \\
&\leq \ell_0^d + |S| (\ell_1^3 \ell_0^{d - 3} - \ell_0^d) + \sum_{v_i \notin S} \vol{(O_{v_i} \setminus O_{u_i}) \setminus B} \\
&= \ell_0^d + |S| (\ell_1^3 \ell_0^{d - 3} - \ell_0^d) + \sum_{v_i \notin S} \left((\ell_1^3 \ell_0^{d - 3} - \ell_0^d) - \vol{(O_{v_i} \cap O_{u_i}) \setminus B } \right).
\end{align*}
Now, notice that, for each $v_i \notin S$, $O_{v_i} \cap O_{u_i} \supset \prod_{e \in E} I_{\bone[e = e_i]}$. Thus, $\vol{(O_{v_i} \cap O_{u_i}) \setminus B} \geq (\ell_1 - \ell_0)\ell_0^{d-1}$. Plugging this into the above, we get
\begin{align*}
&\vol{O_{v_1} \cup \cdots \cup O_{v_k} } \\
&\leq \ell_0^d + |S| (\ell_1^3 \ell_0^{d - 3} - \ell_0^d) + \sum_{v_i \notin S} \left((\ell_1^3 \ell_0^{d - 3} - \ell_0^d) - (\ell_1 - \ell_0)\ell_0^{d-1}\right) \\
&= V - (k - |S|) (\ell_1 - \ell_0)\ell_0^{d-1} \\
&\leq V \left(1 - \frac{k - |S|}{k} \cdot \frac{(\ell_1  / \ell_0) - 1}{(\ell_1 / \ell_0)^3 - 1}\right).
\end{align*}
Note that $\frac{k - |S|}{k} \geq \delta$ and that, when picking $\ell_1 > \ell_0$ as constants, $\frac{(\ell_1  / \ell_0) - 1}{(\ell_1 / \ell_0)^3 - 1} = \Omega(1)$. Thus, the gap between the two cases is $1 - \eps$ for some $\eps = \Omega(\delta)$ as desired.
\end{proof}


\subsubsection{Hardness in Low Dimension}\label{sec:hard:vol:low}

Next, we continue with the low dimension case. For any dimension $d \geq 3$, we are \emph{not} able to rule out PTASes. However, we still manage to prove that $\left(1 - O\left(\frac{1}{\log n}\right)\right)$-approximation is NP-hard. We remark also that, assuming NP $\nsubseteq$ TIME($2^{n^{o(1)}}$), this also rules out PTAS with running time e.g. $n^{2^{O(1/\eps^{0.99})}}$. In other words, while PTASes may still exist for the problem, our result imples (under standard complexity-theoretic assumptions) that their running time must be extremely large. Our result is stated more formally below.


\begin{theorem} \label{thm:main-low-dim-hardness}
For any $d \geq 3$, the max-volume selection problem is NP-hard to approximate to within $\Paren{1 - \frac{\eps}{\log n}}$-factor for some constant $\eps > 0$, even when each object is a unit-volume box.
\end{theorem}

We will prove this via a similar approach to that of hardness of approximation for Maximum Independent Set (Max-IS) on intersection graphs of boxes due to~\cite{ChlebikC05}. In fact, we can show a version of Max-IS on intersection graphs of boxes that has a stronger property that the intersection are either heavy or empty. This is formalized below, and might be of independent interest.

\begin{definition}
A set of objects $O_1, \dots, O_n$ are said to satisfy \emph{$\gamma$-heavy-intersection} property if, for every $O_i, O_j$ such that $O_i \cap O_j \ne \emptyset$, we have $\vol{O_i \cap O_j} \geq \gamma$.
\end{definition}

\begin{theorem} \label{thm:low-dim-heavy-ind}
There exist constants $\gamma, \zeta > 0$ such that the following holds. For any $d \geq 3$, given a set $\{B_1, \dots, B_n\}$ of unit-volume boxes in $\R^d$ with $\gamma$-heavy-intersection property and an integer $k$, it is NP-hard to distinguish between:
\begin{itemize}
\item (YES) There exists $k$ boxes from the input that form an independent set.
\item (NO) There is no $\Paren{1 - \frac{\zeta}{\log n}} k$ boxes from the input that form an independent set.
\end{itemize}
\end{theorem}

Before we prove \Cref{thm:low-dim-heavy-ind}, let us note that it implies our main hardness (\Cref{thm:main-low-dim-hardness}).

\begin{proof}[Proof of \Cref{thm:main-low-dim-hardness}]
Let $\eps = \gamma \cdot \zeta$ where $\gamma, \zeta$ are the constants from \Cref{thm:low-dim-heavy-ind}.

Let $(B_1, \dots, B_n, k)$ be the input from \Cref{thm:low-dim-heavy-ind}. In the YES case, there exist $k$ disjoint boxes. The volume of the union of these $k$ disjoint boxes is exactly $k$.

As for the NO case, suppose that there is no $\Paren{1 - \frac{\zeta}{\log n}} k$ boxes from the input that form an independent set. Consider any $k$ boxes $(B_i)_{i \in I}$. Let $(B_j)_{j \in J}$ where $J \subseteq I$ denote the maximum independent set among these boxes; from our assumption, we have $|J| < \Paren{1 - \frac{\zeta}{\log n}} k$. Moreover, since $(B_j)_{j \in J}$ is a maximum independent set among $(B_i)_{i \in I}$, for any $\ell \in I \setminus J$, $B_\ell$ must intersect with at least one box from $(B_j)_{j \in J}$. Thus, from $\gamma$-heavy-intersection property, we have
\begin{align*}
\vol{\bigcup_{i \in I} B_i} &\leq \vol{\bigcup_{j \in J} B_j} + \sum_{\ell \in I \setminus J} \vol{B_\ell \setminus \bigcup{j \in J} B_j} \\
&\leq |J| + |I \setminus J| \cdot (1 - \gamma) \\
&= k - |I \setminus J| \cdot \gamma \\
&< k \Paren{1 - \frac{\gamma \zeta}{\log n}} = k \Paren{1 - \frac{\eps}{\log n}} \qedhere.
\end{align*}
\end{proof}

We will next prove \Cref{thm:main-low-dim-hardness}. Our proof follows the general approach of~\cite{ChlebikC05}.
In particular, their reduction proceeds by first ``subdividing'' the graph before embedding it in $d$-dimensional boxes. To describe this, let us first define a subdivision of a graph.

\begin{definition}[Subdivision Graph]
Let $G = (V, E)$ be any graph. For $s \in \N$, its $s$-subdivision is the graph $G_s$ obtained from $G$ by $s$-subdivision of each edge $e$, i.e., replacing $e=\{u,v\}$ by a path with end vertices $u$, $v$, and $s$ new internal vertices\footnote{Note that the paths are pairwise internally disjoint.} $e^1, \dots, e^s$.
\end{definition}

Chleb{\'{\i}}k and Chleb{\'{\i}}kov{\'{a}} \cite{ChlebikC05} shows that any $s$-subdivision of a graph where $s \geq 3$ can be embedded as an intersection graph of boxes in $\R^3$. Unfortunately, this is insufficient for us since their box volumes vary significantly. In particular, while each of their ``original vertex'' boxes has unit volume, each of their ``new internal vertex'' boxes has volume as large as $\Omega(n^2)$. This prevents us from using their reduction from our problem.

As stated above, we will instead provide an embedding that satisfies $\gamma$-heavy-intersection property. However, our reduction will need $s$ to be $\Omega(\log n)$. The main lemma we need is the following which allows us to connect any two unit cubes at distance $D$ (along some axis) using $\Omega(\log D)$ intermediate boxes while satisfying $\gamma$-heavy-intersection property. We remark that the dependency on $\log D$ is necessary: To satisfy $\gamma$-heavy-intersection each length of the box cannot change by a factor of more than $1/\gamma$ in each step. Thus, to achieve a distance of $D$, we need $\Omega_\gamma(\log D)$ steps.

Throughout the remainder of the proof, for a center $c \in \R^3$, we use $B^c$ to denote the unit cube centered at $c$, i.e. $B^c = [c_1 - 1/2, c_1 + 1/2] \times [c_2 - 1/2, c_2 + 1/2] \times [c_3 - 1/2, c_3 + 1/2]$. For $i \in [3]$, we also use $\bone_i \in \{0, 1\}^3$ to denote the vector that is one only at the $i$-coordinate. 

\begin{lemma}[Single-Axis Box Path]
\label{lem:internal_boxes}
Let $c, c' \in \R^3$ be two points differing in exactly one coordinate by a distance $D$. That is, $c' = c + D \cdot \bone_i$ for some $i \in [3]$ and $D \in \R_{> 0}$. Let $L \ge 2$ be an integer, and suppose the distance $D$ satisfies $100 \cdot 2^{L/2} \le D \le 2^L$. 

Then, there exists a sequence of $L$ unit-volume, axis-parallel boxes $B_1, B_2, \dots, B_L$ such that, setting $B_0 = B^c$ and $B_{L+1} = B^{c'}$:
\begin{enumerate}
    \item \textbf{Consecutive Intersection:} $\mathrm{Vol}(B_{\ell-1} \cap B_\ell) \ge \gamma$ for all $\ell \in [L + 1]$, where $\gamma = 1/512$.
    \item \textbf{Non-consecutive Disjointness:} $B_{\ell} \cap B_{\ell'} = \emptyset$ for all $|\ell - \ell'| \ge 2$.
    \item \textbf{Strictly Narrow Bounding Tube:} Let $T_{c, c'}$ be the bounding box (or ``tube'') defined by the interval $(c_i + 1/4, c'_i - 1/4)$ on the $i$-coordinate and the intervals $(c_j - 1/5, c_j + 1/5)$ on the other coordinates $j \ne i$. All intermediate boxes $B_1, \dots, B_L$ are contained within $T$.
\end{enumerate}
\end{lemma}

\begin{proof}
Without loss of generality, assume $c = (0,0,0)$ and $c' = (D, 0, 0)$, so the differing coordinate is the first axis. The tube $T$ is therefore defined as $[1/4, D - 1/4] \times [-1/5, 1/5] \times [-1/5, 1/5]$. 

We will give a proof below when $L$ is an even integer. The case where $L$ is odd follows similarly\footnote{For odd $L$, the sequence $X_\ell$'s has a single maximal peak at index $m = (L+1)/2$: We set $X_\ell = 32\alpha^\ell$ for $\ell \le m$ and $X_\ell = 32\alpha^{L+1-\ell}$ for $\ell > m$. The special step $\delta_{k+1} = \frac{3}{4}X_k$ is omitted; instead, we set $\delta_\ell = \max(X_{\ell-1}, X_\ell)/2$ for all $2 \leq \ell \leq L$}.

Each box $B_\ell$ has center $c_\ell = (x_\ell, 0, 0)$, and dimensions $X_\ell \times Y_\ell \times Z_\ell$ such that $X_\ell Y_\ell Z_\ell = 1$. In particular, we set $Y_\ell = Z_\ell = X_\ell^{-1/2}$, where $x_{\ell}$ and $X_\ell$ will be specified below.

Since $L$ is even, we evenly partition the $L$ boxes into a growing phase of $k = L/2$ boxes and a shrinking phase of $k$ boxes. We introduce a growth parameter $\alpha \in [2, 4]$ and set the $x$-dimensions $X_\ell$ as follows:
\begin{itemize}
    \item \textbf{Growing Phase} ($1 \le \ell \le k$): $X_\ell = 32\alpha^\ell$.
    \item \textbf{Shrinking Phase} ($k < \ell \le L$): $X_\ell = 32\alpha^{L+1-\ell}$.
\end{itemize}

We define the center-to-center step sizes $\delta_\ell = x_\ell - x_{\ell-1}$ (where $x_0 = 0$ and $x_{L+1} = D$):
\begin{itemize}
    \item $\delta_1 = X_1/2 + 1/4 = 16\alpha + 1/4$.
    \item $\delta_\ell = X_\ell/2 = 16\alpha^\ell$ for $2 \le \ell \le k$.
    \item $\delta_{k+1} = \frac{3}{4} X_k = 24\alpha^k$.
    \item $\delta_\ell = X_{\ell-1}/2 = 16\alpha^{L+2-\ell}$ for $k+2 \le \ell \le L$.
    \item $\delta_{L+1} = X_L/2 + 1/4 = 16\alpha + 1/4$.
\end{itemize}

Total length covered is a continuous function $D(\alpha) = \sum_{j=1}^{L+1} \delta_j$. Evaluating at $\alpha=2$, the sum evaluates strictly to $D(2) \le 100 \cdot 2^{L/2}$. Evaluating at $\alpha=4$, the geometric series evaluates to $D(4) \ge 12 \cdot 2^L \ge 2^L$. Because $100 \cdot 2^{L/2} \le D \le 2^L$, the Intermediate Value Theorem guarantees there exists $\alpha \in [2, 4]$ such that $D(\alpha) = D$ exactly. We use this value of $\alpha$ in the construction. 

We next prove the three properties claimed in the lemma.

\textbf{Property 3 (Bounding Tube):} Since $\alpha \ge 2$, every intermediate box has $X_\ell \ge 64$. Consequently, $Y_\ell = Z_\ell \le 1/\sqrt{64} = 1/8$, which is strictly less than $1/5$. Thus, they are contained within $[-1/5, 1/5]$ on the last two axes. Along the first axis, the left boundary of $B_\ell$ is $x_\ell - X_\ell/2$ and the right boundary is $x_\ell + X_\ell/2$. Because $\delta_m \ge X_m/2$ for all steps (and symmetry), these intervals are strictly contained within the span of the first and last intermediate boxes: $[x_\ell - X_\ell/2, x_\ell + X_\ell/2] \subseteq [x_1 - X_1/2, x_L + X_L/2] = [1/4, D - 1/4]$.

\textbf{Property 1 (Intersection):} 
Due to symmetry, it suffices to consider $\ell \in [k + 1]$. First, consider the case $\ell = 1$. $B_0$ spans $[-1/2, 1/2]$ in the first axis, while $B_1$ spans $[1/4, X_1 + 1/4]$ where $X_1 > 1/4$. Thus, the first axis overlap is exactly $1/4$. Each of the other two axes overlap is $1/\sqrt{X_1} = 1/\sqrt{32\alpha}$. The intersection volume is $\frac{1}{4} \cdot \frac{1}{32\alpha} = \frac{1}{128\alpha} \ge 1/512$. 

For $2 \leq \ell \leq k$, the first axis overlap is $X_{\ell-1}/2 = 16\alpha^{\ell-1}$. Each of the other two axes overlap is $1/\sqrt{X_\ell} = 1/\sqrt{32\alpha^\ell}$. This yields the intersection volume $\frac{1}{2\alpha} \ge 1/8 \ge 1/512$. 

Finally, at the peak transition ($\ell = k+1$), the first axis overlap is $X_k - \delta_{k+1} = X_k - \frac{3}{4}X_k = \frac{1}{4}X_k$, while each of the other two axes overlap is $1/\sqrt{X_k}$. The intersection volume is thus  $1/4 \ge 1/512$.

\textbf{Property 2 (Disjointness):} 
Consider any $\ell, \ell'$ such that $|\ell - \ell'| \geq 2$. Assume w.l.o.g. that $\ell < \ell'$. The distance between their centers is $x_{\ell'} - x_\ell = \sum_{m=\ell+1}^{\ell'} \delta_m \geq \delta_{\ell + 1} + \delta_{\ell'}$. Since $\delta_{\ell+1} > X_\ell/2$ and $\delta_{\ell'} \geq X_{\ell'} / 2$, we have $x_j - x_i \geq X_i/2 + X_j/2$ and $B_\ell, B_{\ell'}$ are disjoint.
\end{proof}

Using the above axis-aligned box path, we can construct the embedding for $s$-division graph for $s = \Omega(\log n)$ as defined below. The idea is roughly similar to \cite{ChlebikC05}, where, for each vertex, we assign each of its edge to a different coordinate and ``route'' the subdivided edge through that coordinate. The main difference is that we use our box path above for such a routing. 


For routing to go through different coordinates, it will be convenient\footnote{It is possible to construct an embedding without the ``Edge Distinctness'' requirement in the lemma, but this property helps simplify the reduction.} for us to use the following lemma, which is simple to show. Here $A(w, e)$ is the coordinate $w$ will use to route $e$.

\begin{lemma} \label{lem:edge-vertex-coloring}
Let $G = (V, E)$ be any graph with $n$ vertices and maximum degree $d \geq 2$. There exists an assignment $A(v, e) \in [d]$ for every vertex $v$ that is adjacent to edge $e$ that satisfies the following:
\begin{itemize}
\item (Vertex Distinctness) For every vertex $v$, $A(v, e)$ are distinct for all edges $e$ adjacent to $v$.
\item (Edge Distinctness) For every edge $e = \{u, w\}$, $A(u, e)$ and $A(w, e)$ are distinct.
\end{itemize}
Moreover, such an assignment can be found in polynomial time.
\end{lemma}

\begin{proof}
We construct a graph $G' = (V' = V \cup E, E')$. The vertex set $V'$ of $G'$ consists of all vertices and all edges of $G$. The set of edges $E'$ consists of all pairs $\{e, v\}$ of edge $e \in E$ that is adjacent to $v \in V$.
Note that this is a bipartite graph with maximum degree $d$. As a result, there is an proper edge coloring $\chi: E' \to [d]$ with $d$ colors and this can be found in polynomial (in fact nearly linear) time~\cite{ColeOS01}. This gives us the desired assignment; namely, for every adjacent vertex $v$ and edge $e$ pair, we set $A(v, e) = \chi(\{v, e\})$.
\end{proof}

We are now ready to prove the main theorem for the embedding:

\begin{theorem} \label{thm:main-embedding-to-r3}
Let $G = (V, E)$ be any graph with $n$ vertices where each vertex has degree at most 3. Let $s = 6 \cdot L + 6$ where $L = 4 \lceil \log_2 n \rceil + 20$, and let $G_s$ denote the $s$-subdivision of $G$. Then it is possible to represent each vertex in $G_s$ as a unit-volume $3$-dimensional box such that two boxes intersect if and only if they are neighbors in $G_s$, and these boxes satisfy the $\gamma$-heavy-intersection property with $\gamma = 1/512$.
\end{theorem}

\begin{proof}
Let $V = \{v_1, \dots, v_n\}$ and $E = \{e_1, \dots, e_m\}$. 
We map every vertex in $G_s$ to a unit-volume box in $\R^3$ as follows.

\paragraph{Original Vertex Placement.}
Let $K = 2 \cdot 10^6 n^3$. First, we assign each original vertex $v_i \in V$ a unit cube $B^{C_i}$ centered at $C_i = (Ki, Ki, Ki)$. 


\paragraph{Edge Routing.}
For each edge $e = e_j = (v_u, v_w) \in E$ with $u < w$, we will now describe the $s$ boxes we use to represent its subdivided vertices in $G_s$.

First, we assign a unique spatial ``highway'' coordinate $H_j = K(n + j)$. Let $d_1 = A(v_u, e) \in [3], d_3 = A(v_w, e) \in [3]$ be the assignments guaranteed from \Cref{lem:edge-vertex-coloring}. Note that, by \Cref{lem:edge-vertex-coloring}, $d_1, d_3$ are distinct. Finally, let $d_2$ be the remaining element in $[3] \setminus \{d_1, d_3\}$.

The subdivided path for $e_j$ is routed as an orthogonal polyline defined by the centers $C_u$, $C_w$, and 5 intermediate unit cubes centered points $P^{e, u, 1}, P^{e, u, 2}, P^e, P^{e, w, 2}, P^{e, w, 1}$, defined as follows.
\begin{itemize}
    \item $P^{e, u, 1} = C_u + (H_j - Ku)\bone_{d_1}$
    \item $P^{e, u, 2} = P^{e, u, 1} + (H_j - Ku)\bone_{d_2}$
    \item $P^e = P^{e, u, 2} + (H_j - Ku)\bone_{d_3} = (H_{j,u}, H_{j,u}, H_{j,u}) = P^{e, w, 2} + (H_j - Kw)\bone_{d_1}$
    \item $P^{e, w, 2} = P^{e, w, 1} + (H_j - Kw)\bone_{d_2}$
    \item $P^{e, w, 1} = C_w + (H_j - Kw)\bone_{d_3}$
\end{itemize}
This sequence guarantees exactly 6 straight, axis-parallel segments. We refer to these centers (including the centers $C_u, C_w$) as the ``corner centers''.

We used 5 boxes for the intermediate corners, leaving $s - 5 = 6L + 1$ boxes. We assign $L$ boxes to each of 5 out of the 6 segments, and assign $L + 1$ boxes to the remaining segment. By construction, the length $D$ of any segment is at least $K \geq 10^6 n^2$ and at most $3 K n \leq 6 \cdot 10^6 n^3$. This ensures that $100 \cdot 2^{L/2} \le D \le 2^L$ for all sufficiently large $n$. We can thus apply \Cref{lem:internal_boxes} to embed the $L$ or $L + 1$ intermediate boxes into each of the 6 segments. The boxes satisfy the internal adjacency of the $G_s$ path, overlapping only with their direct sequence neighbors with intersection volume $\ge 1/512$.

Finally, we need to argue that the boxes belonging to different segments or edges never intersect.
%
By Lemma \ref{lem:internal_boxes}, the intermediate boxes of a segment (of centers) spanning from $P$ to $\tP$ along axis $\omega$ are strictly confined to the tube $T(P, \tP)$ defined by $[\min(P_{\omega}, \tP_{\omega}) + 1/4, \max(P_{\omega}, \tP_{\omega}) - 1/4]$ on axis $\omega$, and $[P_{\nu} - 1/5, \tP_{\nu} + 1/5]$ on the two transverse axes $\nu \ne \omega$. Thus, it suffices to show that $T(P, \tP) \cap T(P', \tP') = \emptyset$ for any pair of distinct segments $(P, \tP)$ and $(P', \tP')$.

To prove this, we consider two cases based on whether the two segments share an endpoint:
\begin{itemize}
\item Case I: The two segments share an endpoint. In other words, $\{P, \tP\} \cap \{P', \tP'\}$ is non-empty. Let us assume w.l.o.g. that $P = P'$. Observe from our construction that the segments $(P, \tP)$ and $(P, \tP')$ must be along different axis $\omega \ne \omega'$. In the $\omega$ axis, we have that $T(P, \tP)$ is fully contained in $[\min(P_{\omega}, \tP_{\omega}) + 1/4, \max(P_{\omega}, \tP_{\omega}) - 1/4]$, while $T(P, \tP')$ is fully contained in $[P_{\omega} - 1/5, P_{\omega} + 1/5]$. These two intervals are disjoint since $|\tP_{\omega} - P_{\omega}| \geq 1$.
\item Case II: The two segments do not share an endpoint. Suppose for the sake of contradiction that $P, \tP, P', \tP'$ are all distinct but $T(P, \tP) \cap T(P', \tP') \ne \emptyset$. Let $\omega$ and $\omega'$ are the axes corresponding to the segments $(P, \tP)$ and $(P', \tP')$, respectively. Finally, let $\nu \in [3] \setminus \{\omega, \omega'\}$ denote the remaining axis. Notice that along $\nu$ the bounding boxes $T(P, \tP)$ and $T(P', \tP')$ are bounded in $[P_{\nu} - 1/5, P_{\nu} + 1/5]$ and $[P'_{\nu} - 1/5, P'_{\nu} + 1/5]$ respectively. Since both $P_{\nu}$ and $P'_{\nu}$ are integers, we must have $P_{\nu} = P'_{\nu}$ for the bounding boxes to overlap. In other words, we must have $P_{\nu} = \tP_{\nu} = P'_{\nu} = \tP'_{\nu}$, i.e. there are four distinct intermediate corner boxes with the same value on the $\nu$ axis.

Observe from our construction that this is impossible if $P_{\nu} = H_j$ for some $j \in [m]$; since in the sequence of our corner box construction for $e_j$, any coordinate is always changed after 3 steps and, thus, there are at most three distinct boxes with coordinate $\nu$ is equal to $H_j$. Hence, we must have $P_{\nu} = Ki$ for some $i \in [n]$.

Now, consider the vertex $v_i$. Note that this immediately implies that there are two edges $e, e'$ such that $A(v_i, e), A(v_i, e') \ne \nu$ (as otherwise there would only be at most three corner centers whose $\nu$-coordinate is equal to $Ki$). Thus, there are at most five corner centers with $\nu$-coordinates equal to $Ki$: $C_i, P^{e, u, 1}, P^{e, u, 2}, P^{e', u, 1}, P^{e', u, 2}$. Since the two segments do not share a corner, one of the segment must be $(P^{e, u, 1}, P^{e, u, 2})$ or $(P^{e', u, 1}, P^{e', u, 2})$. Assume w.l.o.g. that it is the former case, i.e. that $P = P^{e, u, 1}, \tP = P^{e, u, 2}$.

Let $e = e_j$ and $e' = e_{j'}$ for some $j, j' \in [m]$. Consider the coordinate $\rho = A(v_i, e)$. We have that $P_\rho = \tP_\rho = H_j$. Thus, the bounding box $T(P, \tP)$ is entirely contained in $[H_j - 1/5, H_j + 1/5]$ in the $\rho$ coordinate. For $T(P', \tP')$ to intersect $T(P, \tP)$, we must have that $\{P'_{\rho}, \tP'_{\rho}\} = \{Ki, H_{j'}\}$ and $j' > j$. Now, consider the remaining coordinate $\rho'$ (apart from $\rho, \nu$). We must have $P'_{\rho'} = \tP'_{\rho'}$ is either $Ki$ or $H_{j'}$. This means that, on the $\rho'$ axis, $T(P', \tP')$ is either entirely contained in $[Ki - 1/5, Ki + 1/5]$ or $[H_{j'} - 1/5, H_{j'} + 1/5]$. Both of these are disjoint from $[Ki + 1/4, H_j - 1/4]$. Meanwhile, we have $\{P_{\rho'}, \tP_{\rho'}\} = \{Ki, H_j\}$, meaning that, on the $\rho'$ axis, $T(P, \tP)$ is entirely contained in $[Ki + 1/4, H_j - 1/4]$. Thus, $T(P, \tP)$ and $T(P', \tP')$ are disjoint, a contradiction. \qedhere
\end{itemize}
\end{proof}

With the above embedding in mind, we are nearly done. We simply recall the hardness of approximation of Max-IS on degree-3 graphs and how the gap transfers on its $s$-division below. Here we use $\optis(G)$ to denote the size of the maximum independent set of $G$.

\begin{theorem}[\cite{ChlebikC03}] \label{thm:hardness-is-deg3}
There exists a constant $\eps > 0$ such that the following holds. Given a graph $G$ with maximum degree 3 and an integer $k$, it is NP-hard to distinguish between:
\begin{itemize}
\item (YES) $\optis(G) \geq k$.
\item (NO) $\optis(G) < (1 - \eps) k$.
\end{itemize}
\end{theorem}

\begin{theorem}[\cite{ChlebikC05}] \label{thm:subdivision-opt}
Let $G=(V,E)$ be any graph, let $s$ be any nonnegative even integer, and let $G_s$ be the $s$-subdivision of $G$. Then, $\optis(G_s) = \optis(G) + |E| \cdot \frac{s}{2}$.
\end{theorem}

We can now prove \Cref{thm:main-low-dim-hardness}.

\begin{proof}[Proof of \Cref{thm:main-low-dim-hardness}]
Let $\delta = \eps/2000$ where $\eps$ is the constant from \Cref{thm:hardness-is-deg3}. 
For any input $(G = (V, E), k)$ from \Cref{thm:hardness-is-deg3}, apply \Cref{thm:main-embedding-to-r3} to produce its $s$-subdivision $G_s = (V_s, E_s)$ where $s = 6 \cdot \Paren{4 \lceil \log_2 |V| \rceil + 20} + 6$ together with its box $\gamma$-heavy-intersection embedding $(B_s)_{G_s}$ such that its intersection graph is exactly $G_s$. Finally, let $k' = k + s \cdot |E|/2$. It is obvious that this reduction runs in polynomial time. We next prove its soundness and completeness.

\paragraph{(Completeness)} If $\optis(G) \geq k$, \Cref{thm:subdivision-opt} implies that $\optis(G_s) \geq k'$ as desired.

\paragraph{(Soundness)} Suppose that $\optis(G) < (1 - \eps) k$. Then, \Cref{thm:subdivision-opt} yields $\optis(G_s) < k' - \eps \cdot k$. Notice that $k \geq |V|/4$ because $G$ has maximum degree 3. Furthermore, $n = |V_s|$ is exactly equal to $|V| + |E| \cdot s \leq 1000 |V| \log |V|$. Thus, we have
\begin{align*}
\optis(G_s) < k'\Paren{1 - \eps \cdot \frac{k}{k'}} \leq k' \Paren{1 - \eps \cdot \frac{|V|/3}{1000 |V| \log |V|}} \leq k' \Paren{1 - \frac{\delta}{\log n}}. &\qedhere
\end{align*}
\end{proof}

\subsection{Max Coverage}

We next consider the original (discrete) version of the max coverage problem. As discussed earlier, any algorithm for max coverage can be used for max-volume selection with the approximation ratio nearly preserved. As such all the hardness results in the previous subsection immediately translates to the max coverage problem as well. Below, we will show additional hardness results for max coverage.

\subsubsection{Hardness in High Dimension}\label{sec:hard:high}

In the low dimensional case, we rule out the existence of PTAS for many classes of objects:

\begin{theorem} \label{thm:apx-hardness-max-cov}
Assuming $\mbox{P} \ne \mbox{NP}$, there is no PTAS for
the max coverage problem for each of the following classes of objects:
\begin{itemize}
    \item Axis-aligned rectangles in $\mathbb{R}^2$ even when all rectangles have lower-left corner in $[-1,-1+\epsilon]\times[-1,-1+\epsilon]$ and upper-right corner in $[1,1+\epsilon]\times[1,1+\epsilon]$ for an arbitrarily small $\epsilon>0$. 
    \item Axis-aligned ellipses in $\mathbb{R}^2$, even when all ellipses have centers in $[0,\epsilon]\times[0,\epsilon]$ and major and minor axes of length in $[1,1+\epsilon]$. 
    \item Axis-aligned slabs in $\mathbb{R}^2$, each of the form $[a_i,b_i]\times[-\infty,\infty]$ or $[-\infty,\infty]\times[a_i,b_i]$. 
    \item Axis-aligned rectangles in $\mathbb{R}^2$, even when the boundaries of each pair of rectangles intersect exactly zero times or four times. 
    \item Downward shadows of line segments in $\mathbb{R}^2$. 
    \item Downward shadows of (graphs of) univariate cubic functions in $\mathbb{R}^2$. 
    \item Unit balls in $\mathbb{R}^3$, even when all the balls contain a common point. 
    \item Axis-aligned cubes in $\mathbb{R}^3$, even when all the cubes contain a common point and are of similar size. 
    \item Half-spaces in $\mathbb{R}^4$. 
    \item Fat semi-infinite wedges in $\mathbb{R}^2$, each of which has an opening angle in $[\pi-\epsilon,\pi)$ and has its vertex within $\epsilon$ of a common point. 
\end{itemize}
\end{theorem}

Our results follow from a rather direct implication from the geometric set cover problem, where we are given a universe $\cU$ of points and a collection $S \subseteq \cS$ of objects. The goal is to select $I \subseteq S$ that whose union contains all the points. The relation between the Set Cover problem and the Max Cover problem can be stated as follows.

\begin{lemma} \label{lem:maxcov-vs-setcov}
Suppose that there exists a PTAS for max coverage for a class $\cS$. Then, there is also a PTAS for geometric set cover for the same class $\cS$ for all set system $(\cU, S)$ where each set in $S$ has size at most $O(1)$.
\end{lemma}

\begin{proof}
Let $\alg$ be the PTAS for max coverage for a class $\cS$. We construct a PTAS for geometric set cover as follows. Suppose that we are given a constant $\eps > 0$ and a set system $(\cU, S)$ where each set in $S$ has size at most $c$, and assume w.l.o.g.\footnote{Otherwise, we can simply loop through all possible values of $k$ and run the algorithm on each value of $k$.} that we know the optimum set cover size $k$. Then, we run $\alg$ on $(\cU, S)$ to find $k$ sets that cover at least $(1 - \delta) |\cU|$ elements for $\delta = \eps / c$. For the remaining elements $\delta \cdot |\cU|$ elements, we simply pick any set to cover each of them. In total, the number of sets picked is at most
\begin{align*}
k + \delta \cdot |\cU| \leq k + \delta \cdot c \cdot k = (1 + \eps) k,
\end{align*}
where the first inequality follows from the fact that each set has size at most $c$. Thus, this yields a $(1 + \eps)$-approximation as desired.
\end{proof}

Chan and Grant~\cite{ChanG14} proved the following hardness of approximation for geometric set cover with bounded set size.

\begin{theorem}[\cite{ChanG14}]
Geometric set cover is APX-hard for each of the classes of objects from \Cref{thm:apx-hardness-max-cov}. Furthermore, this holds even for set systems $(\cU, S)$ where each set has size at most 3.
\end{theorem}

The above theorem together with \Cref{lem:maxcov-vs-setcov} immediately yields \Cref{thm:apx-hardness-max-cov}.  
This implies APX-hardness for most of the types of objects listed in the third and fourth row of Tables~\ref{tbl:set:cov}--\ref{tbl:max:cov}.
The case of similar-size fat triangles can be handled
by an easy modification of the reduction  from \cite{ChanG14} for downward shadows of 2D line segments.
(For the cases of 4D unit hypercubes and 4D boxes with the origin as a vertex, there is an easy reduction from the case of 2D rectangles.)

\subsubsection{Hardness in Low Dimension}\label{sec:hard:low}

In high dimension, we prove a hardness that is arbitrarily close to the $(1 - 1/e)$ upper bound, even when the objects are boxes. Here, the number of dimensions required is a constant that depends on how close we wish the hardness to be to the optimal.

\begin{theorem} \label{thm:tight-hardness-max-cov}
For every $\eps > 0$, there is a sufficiently large $d > 0$ such that it is NP-hard to approximate max coverage to within a factor of $(1 - 1/e + \eps)$, even when the objects are unit boxes in $\R^d$.
\end{theorem}

We use the following classic result of Feige on the hardness of the max coverage problem\footnote{We remark that, while Feige did not explicitly state the bounded degree and frequency properties, it is simple to check that this hold in the construction.}: 

\begin{theorem}[\cite{Feige98}] \label{thm:feige-bounded}
For every $\eps > 0$, there are sufficiently large $c, f > 0$ such that it is NP-hard to approximate max coverage to within a factor of $(1 - 1/e + \eps)$, even when each set has size at most $c$ and each element appears in at most $f$ sets.
\end{theorem}

\begin{proof}[Proof of \Cref{thm:tight-hardness-max-cov}]
From \Cref{thm:feige-bounded}, it suffices to show that any set system $(\cU, S)$ where each set has size at most $c$ and each element appears in at most $f$ sets can be represented as points and unit boxes in $\R^d$ where $d = 2cf$. This embedding can be constructed as follows:
\begin{itemize}
\item Assume w.l.o.g. that $\cU = [m]$.
\item First, let $\Delta = cf$ and determine a coloring $\phi: [m] \to [\Delta]$ such that no two elements that are in a common set have the same color. This always exists and can be computed greedily since, for each element, the number of other elements that belong to the same set to it is at most $f(c - 1) < \Delta$.
\item Now, for each element $i \in [m]$, we construct the corresponding point $p^i \in \R^d$ as follows:
\begin{align*}
(p^i)_j =
\begin{cases}
\frac{1}{2}\left(1 + \frac{i}{m}\right) & \text{ if } j = 2 \cdot \phi(i) - 1 \\
\frac{1}{2}\left(1 - \frac{i}{m}\right) & \text{ if } j = 2 \cdot \phi(i) \\
0 &\text{ Otherwise. }
\end{cases}
\end{align*}
\item Finally, for each set $s \in S$, we let the box $B^s$ be the unit box with center $c^s \in \R^d$ be as follows:
\begin{align*}
(c^s)_j =
\begin{cases}
\frac{1}{2} \cdot \frac{i}{m} & \text{ if } j = 2 \cdot \phi(i) - 1 \text{ for some } i \in s \\
-\frac{1}{2} \cdot \frac{i}{m}& \text{ if } j = 2 \cdot \phi(i) \text{ for some } i \in s \\
0 &\text{ Otherwise. }
\end{cases}
\end{align*}
We remark that, when $j = 2\phi(i) - 1$ or $j = 2\phi(i)$ for some $i \in s$ (the first two cases), the value of $i$ is unique due to our choice of the coloring $\phi$.
\end{itemize}
We claim that this indeed encodes the set system. To see this, consider any $i \in [m]$ and $s \in S$.
\begin{itemize}
\item \textbf{Case I:} $i \in s$. In this case, it is clear from the definition that $(p^i)_j \in \left[(c^s)_j - \frac{1}{2}, (c^s)_j + \frac{1}{2}\right]$ for all $j \in [d]$. Thus, $p^i$ belongs to $B^s$ as desired. 
\item \textbf{Case II:} $i \notin s$. We further consider three subcases:
\begin{itemize}
\item \textbf{Case II.A:} There is no $i' \in s$ such that $\phi(i') = \phi(i)$. In this case, for $j = 2\phi(i) - 1$, we have $(p^i)_j = \frac{1}{2}\left(1 + \frac{i}{m}\right) > \frac{1}{2}$ and $c^s_{j} = 0$, implying that $p^i \notin B^s$.
\item \textbf{Case II.B:} There is some $i' \in s$ such that $\phi(i') = \phi(i)$ and $i' < i$. In this case, for $j = 2\phi(i) - 1$, we have $(p^i)_j = \frac{1}{2}\left(1 + \frac{i}{m}\right) > \frac{1}{2}\left(1 + \frac{i'}{m}\right)$ and $c^s_{j} = \frac{1}{2} \cdot \frac{i'}{m}$, implying that $p^i \notin B^s$.
\item \textbf{Case II.C:} There is some $i' \in s$ such that $\phi(i') = \phi(i)$ and $i' > i$. In this case, for $j = 2\phi(i)$, we have $(p^i)_j = \frac{1}{2}\left(1 - \frac{i}{m}\right) > \frac{1}{2}\left(1 - \frac{i'}{m}\right)$ and $c^s_{j} = -\frac{1}{2} \cdot \frac{i'}{m}$, implying that $p^i \notin B^s$.
\end{itemize}
\end{itemize}
Thus, these points and boxes encode the set system $(\cU, S)$, which concludes our proof.
\end{proof}

\section*{AI Disclosure Statement}

We use Gemini 3.1 Pro to help draft the proofs of Lemma~\ref{lem:internal_boxes} and Theorem~\ref{thm:main-embedding-to-r3}. In particular, we prompt it with the proof sketches and ask it to help flesh out the proofs, including determining the exact constants required. We then edit the generated proofs ourselves to simplify the construction and arguments further. Apart from these two proofs, all other proofs are entirely written by humans.

{\small
\bibliographystyle{alphaurl}
\bibliography{ref}
}

\appendix

\section{LP-Based Approximation Algorithm for Discrete Independent Set\\ (Proof of Lemma~\ref{lem:DIS})}\label{sec:DIS}

\PROOFDIS

\section{LP-Based Approximation Algorithm for Max-Volume Selection}\label{sec:max:cov:var}

In this section, we describe a variant of the max coverage algorithm
from Section~\ref{sec:max:cov:lp} that is more amenable to the
continuous version of the problem (max-volume selection),
where $P=\R^d$.  In the continuous case, for a region $A$, we should re-interpret $|A|$ as the volume of $A$, and sums over $P$ as integrals.

The main issue is that
our earlier algorithm requires invoking a DIS algorithm for
a subset $\Pgood$ of $P$, which is no longer all of $\R^d$.
Instead of assuming a DIS approximation algorithm, we will assume a different property:



\begin{definition}
We say that $\SSS$ satisfies the \emph{$c$-large-DIS property} w.r.t.\ $P$ if for any set $S\subseteq\SSS$,
we can find a subset $I\subseteq S$ such that every point of $P$ lies in at most one object of $I$, and
$\left|\bigcup_{s\in I}s \cap P\right| \ge \left|\bigcup_{s\in S}s \cap P\right|/c$ in polynomial time.
In the case $P=\R^d$, we say that $\SSS$ satisfies the \emph{$c$-large-IS property}.
\end{definition}

\begin{definition}
An object $s$ is \emph{$c$-fat} relative to a box $B$ if $s$ is contained in a homothet of $B$ whose volume is 
at most $c$ times the volume of $s$.

(Note that any convex object in a constant dimension is $O(1)$-fat relative to some box, e.g., by taking a bounding box of the John ellipsoid.)
\end{definition}

\begin{fact}\

\begin{enumerate}
\item[(a)] The family of all $c$-fat objects relative to a fixed box $B$ in $\R^d$ 
satisfies the $O_{d,c}(1)$-large-IS property.
\item[(b)] If $\SSS_1,\ldots,\SSS_\ell$ satisfies the $c$-large-IS property, then
$\SSS_1\cup\cdots\cup\SSS_\ell$ satisfies the $(c\ell)$-large-IS property.
\end{enumerate}
\end{fact}
\begin{proof}
For (a), use a greedy algorithm: add the largest-volume object $s\in S$ to $I$, remove all objects intersecting $s$ from $S$, and repeat.
For each object $s\in I$, the union of all objects that intersect $s$ and have smaller volume than $s$ has volume at most
$O_{d,c}(1)$ times the volume of $s$.    Thus, the volume of $\bigcup_{s\in S}s$ is at most $O_{d,c}(1)$ times the volume of
$\bigcup_{s\in I}s$.

(b) is obvious from the definition.
\end{proof}

\begin{theorem}
Fix $P\subseteq\UUU$.  Suppose $\SSS$ satisfies the $C$-large-DIS property w.r.t.\ $P$.

Then for any $S\subseteq\SSS$, we can find a feasible solution
to the weighted max coverage problem of value at least $(1-1/e+\Omega(1/C^2))\OPTLP-\vmax$ in polynomial time,
where $\OPTLP$ denotes the optimal LP value for weighted max coverage and $\vmax:=\max_{s\in S}|s\cap P|$.  The term $\vmax$ can be removed for the max $k$-coverage problem (where all the weights are equal to $1/k$).
\end{theorem}
\begin{proof}
We modify the algorithm in the proof of Theorem~\ref{thm:maxcov:LP}:
we reset $\Pbad=\emptyset$ and $\Pgood=P$ (so LP rounding in line~1 may be skipped entirely), and in line~4, we compute a feasible
solution $I$ to DIS for $\Sgood$ and $P$ with $\left|\bigcup_{s\in I}s \cap P\right| \ge \left|\bigcup_{s\in\Sgood}s \cap P\right|/C$.

\newcommand{\Ugood}{U_{\text{good}}}

For the analysis of line~4, let $\Ugood=\bigcup_{s\in\Sgood} s\cap P$.  Observe that
\begin{eqnarray*}
\OPTLP &\le& \sum_{p\in\Ugood\cup A_h} z_p' + \sum_{p\in P\setminus\Ugood\setminus A_h} \sum_{s\in\Sbad:\,p\in s} x_s\\
&\le& |\Ugood| + |A_h| + \sum_{s\in\Sbad}\sum_{p\in P\setminus A_h:\,p\in s} x_s\\
&=& |\Ugood| + |A_h| + \sum_{s\in\Sbad} |s\cap P\setminus A_h| x_s\\
&\le& |\Ugood| + |A_h| + (1-\delta')t_h \sum_{s\in\Sbad} w_sx_s
\ \le\ |\Ugood| + |A_h| + (1-\delta')t_h/2,
\end{eqnarray*}
implying that $|\Ugood|\ge (1+\delta')t_h/2$.
Thus, $\left|\bigcup_{s\in I}s \cap P\right|\ge t_h/(2C)$.
Recall that $|s\cap P\setminus A_h|/w_s\le \Delta_h < (1+\delta)t_h$ for every $s\in S$, and $|A_h|\le \delta\,\OPTLP\le 2\delta t_h$.  It follows that
\[ (1+\delta)t_h \sum_{s\in I} w_s\ \ge\ \sum_{s\in I} |s\cap P\setminus A_h|
\ \ge\ \left|\bigcup_{s\in I}s\cap P\right| - |A_h|\ \ge\ \left(\frac{1}{2C}-2\delta\right) t_h,
\]
implying $\sum_{s\in I}w_s\ge 1/(2C)-O(\delta)$.
The rest of the analysis is as before.
\end{proof}

We modify Lemma~\ref{lem:vmax} to remove the $\vmax$ term for the weighted problem:

\begin{lemma}
Fix $P\subseteq\UUU$.
Suppose that for any given $S\subseteq\SSS$, we can find a feasible solution to the 
weighted max coverage problem of value at least $(1-1/e+\delta)\OPT-\vmax$ in polynomial time,
where $\OPT$ denotes the optimal value  and $\vmax:=\max_{s\in S}|s\cap P|$.

Then for any given $S\subseteq\SSS$, we can find a feasible solution to the 
weighted max coverage problem of value at least $(1-1/e+\Omega(\delta))\OPT$ in polynomial time.
\end{lemma}
\begin{proof}
We guess the object $a^*$ in the optimal solution that maximizes $|a^*\cap P|$.
We solve the problem for $\{s\in S: |s\cap P|\le |a^*\cap P|\}$ and $P$ using the given algorithm.
In addition,
we solve the problem for $S$ and $P\setminus a^*$ with budget $1-w_{a^*}$ by a known $(1-e^{-1})$-approximation algorithm, and add back $a^*$ to the solution.
We return the best solution found.

Let  $\vmax := |a^*\cap P|$.
If $\vmax\le (\delta/2)\OPT$, then
the given algorithm yields a solution with value at least 
$(1-e^{-1}+\delta)\OPT - \vmax\ge (1-e^{-1}+\delta/2)\OPT$.
On the other hand, if $\vmax > (\delta/2)\OPT$, then
the solution for the correct guess has value at least
$|a^*\cap P| + (1-e^{-1})(\OPT-|a^*\cap P|) = (1-e^{-1})\OPT + e^{-1}\vmax > (1-e^{-1}+\Omega(\delta))\OPT$.
\end{proof}

Putting everything together, we obtain:

\begin{corollary}
For the following families of objects,
there are polynomial-time $(1-1/e+\Omega(1))$-approximation algorithms for the weighted max-volume selection problem:
\begin{itemize}
\item
objects that are $O(1)$-fat relative to $O(1)$ boxes in any constant dimension $d$---in particular, homothets of $O(1)$ convex objects in any constant dimension $d$.
\end{itemize}
\end{corollary}


\section{Proof of Fact~\ref{fact:expansion}}\label{app:expansion}

We begin with the following lemma (which is related to \emph{Hardy--Littlewood maximal operators}~\cite{SteinBOOK}):
\begin{lemma}\label{lem:extension}\
\begin{enumerate}
\item[(i)] For a line segment $s$, let $(1+\eps)s$ denote the $(1+\eps)$-factor expansion of $s$ (preserving the midpoint).
For any set $A\subseteq\R^d$ and vector $v\in\R^d$, define
\[ \ext_\eps(A,v) = \bigcup_{\text{\em line segment}\ s\subseteq A\ \text{\em parallel to $v$}} (1+\eps)s.
\]
Then $\vol{\ext_\eps(A,v)} \le (1+\eps)\vol{A}$.
\item[(ii)] Let a parallelepiped $\tau$, let $(1+\eps)\tau$ denote the $(1+\eps)$-factor expansion of $\tau$ (preserving the center).  For any set $A\subseteq\R^d$ and parallelepiped $\tau_0\subseteq\R^d$, define
\[ \ext_\eps(A,\tau_0) = \bigcup_{\text{\em parallelepiped}\ \tau\subseteq A\ \text{\em homothetic to $\tau_0$}} (1+\eps)\tau.
\]
Then $\vol{\ext_\eps(A,\tau_0)} \le (1+\eps)^d\vol{A}$.
\end{enumerate}
\end{lemma}
\begin{proof}
First note that (i) is trivial in the $d=1$ case (think of $A$ as a union of disjoint intervals on the real line).
Then (i) for larger $d$ follows by integrating over all lines parallel to $v$.

(ii) follows by applying (i) $d$ times: for a parallelelepiped $\tau_0=\{\xi_1v_1+\cdots+\xi_dv_d:
-1\le \xi_1,\ldots,\xi_d\le 1\}$ defined by vectors $v_1,\ldots,v_d\in\R^d$, we have
\[ \ext_\eps(A,\tau_0)\subseteq \ext_\eps(\cdots \ext_\eps(\ext_\eps(A,v_1),v_2),\ldots, v_d).\qedhere
\]
\end{proof}

Pick a set $V$ of $O_{d,\alpha}(1)$ unit vectors in $\R^d$ such that every vector in $\R^d$ has angle at most $1/c_{d,\alpha}$
from at least one vector in $V$ for a sufficiently large constant $c_{d,\alpha}$.
Let $T$ be the set of all $O_{d,\alpha}(1)$ parallelepipeds defined by $d$ vectors in $V$ centered at the origin.
We claim that
\[ \bigcup_{s\in S} (s\oplus B_{\eps\,\diam(s)})\ \subseteq\ \bigcup_{\tau_0\in T} \ext_{C_{d,\alpha}\eps}\left(\bigcup_{s\in S} s, \tau_0\right)
\]
for a sufficiently large constant $C_{d,\alpha}$.

To see this, let $s\in S$ of diameter $r$ and $p$ be a point in $s$.  Let $B$ be a ball contained in $s$ of radius $\alpha r$, and let $B'$ be the ball of radius $\alpha r/4$ contained in $B$ that is of distance $\Omega(\alpha)$ from $p$.
The set of all directions along which the ray from $p$ intersects $B'$ forms a cone of angle $\Theta_{d,\alpha}(1)$.
We can find $d$ vectors $v_1,\ldots,v_d\in V$ inside this cone, such that
every pair of vectors has angle at least $\Omega_{d,\alpha}(1)$.
Let $\tau_0\in T$ be the parallelepiped defined by $v_1,\ldots,v_d$ centered at the origin.
Then $s$ contains
a homothetic copy $\tau$ of $\tau_0$ with a vertex at $p$ and a scaling factor $r/c'_{d,\alpha}$ for a sufficiently large constant $c'_{d,\alpha}$.
Furthermore, $\tau_0$ is fat since the minimum angle among $v_1,\ldots,v_d$ is $\Omega_{d,\alpha}(1)$, and
so $c''_{d,\alpha}r\tau_0$ contains $B_{\eps r}$ for a sufficiently large constant $c''_{d,\alpha}$.
Thus, $p\oplus B_{\eps r}$ is contained in $(1+C_{d,\alpha}\eps)\tau$ for a sufficiently large constant $C_{d,\alpha}$.
The claim follows.

We conclude that
\begin{eqnarray*}
 \vol{\bigcup_{s\in S} (s\oplus B_{\eps\,\diam(s)}) \setminus \bigcup_{s\in S} s}  &\le&
\sum_{\tau_0\in T} \vol{\ext_{C_{d,\alpha}\eps}\left(\bigcup_{s\in S} s, \tau_0\right) \setminus \bigcup_{s\in S} s}\\
&\le & \sum_{\tau_0\in T} ((1+C_{d,\alpha}\eps)^d-1)\vol{\bigcup_{s\in S}s} \ \le\ O_{d,\alpha}(\eps) \vol{\bigcup_{s\in S}s}.
\end{eqnarray*}

\hfill\qed

\end{document}